\documentclass{vldb}

\usepackage{url}
\usepackage[labelfont=bf,textfont=bf]{caption}
\usepackage{subcaption}
\usepackage{graphicx}
\usepackage{xcolor}
\usepackage{xspace}
\usepackage{array}
\usepackage[noend]{algpseudocode}
\usepackage{dblfloatfix}
\usepackage{balance}
\usepackage{algorithm}
\usepackage{mathtools}
\usepackage{amsmath}

\newcommand{\reffigure}[1]{Figure~\ref{#1}}
\newcommand{\refsection}[1]{Section~\ref{#1}}

\newcommand{\reftable}[1]{Table~\ref{#1}}

\newcommand{\refappendix}[1]{Appendix}

\newcommand{\zap}[1]{}
\newcommand{\btree}{B$^+$-tree\xspace}
\newcommand{\btrees}{B$^+$-trees\xspace}

\newtheorem{theorem}{Theorem}

\newcommand{\longonly}[1]{\empty{#1}}
\newcommand{\shortonly}[1]{\empty{}}
\newcommand{\shortlong}[2]{\empty{#2}}

\newcommand{\scale}{1}

\vldbTitle{On Performance Stability in LSM-based Storage Systems}
\vldbAuthors{Chen Luo, Michael J. Carey}
\vldbDOI{https://doi.org/10.14778/3372716.3372719}
\vldbVolume{13}
\vldbNumber{4}
\vldbYear{2019}

\begin{document}
	
\title{On Performance Stability in LSM-based Storage Systems (Extended Version)}

\numberofauthors{2}

\author{
	\alignauthor
	Chen Luo\\
	\affaddr{University of California, Irvine}\\
	\email{cluo8@uci.edu}
	\alignauthor
	Michael J. Carey\\
	\affaddr{University of California, Irvine}\\
	\email{mjcarey@ics.uci.edu}\\
}
\maketitle

\begin{abstract}
The Log-Structured Merge-Tree (LSM-tree) has been widely adopted for use in modern NoSQL systems for its superior write performance.
Despite the popularity of LSM-trees, they have been criticized for suffering from write stalls and large performance variances
due to the inherent mismatch between their fast in-memory writes and slow background I/O operations.
In this paper, we use a simple yet effective two-phase experimental approach to
evaluate write stalls for various LSM-tree designs.
We further identify and explore the design choices of LSM merge schedulers to minimize write stalls given an {I/O} bandwidth budget.
We have conducted extensive experiments in the context of the Apache AsterixDB system and we present the results here. 
\end{abstract}

\section{Introduction}
In recent years, the Log-Structured Merge-Tree (LSM-tree)~\cite{lsm1996,lsm-survey} has been
widely used in modern key-value stores and NoSQL systems~\cite{cassandra, hbase, leveldb, rocksdb, asterixdb2014}.
Different from traditional index structures, such as \btrees, which apply updates in-place, an LSM-tree always buffers writes into memory.
When memory is full, writes are flushed to disk and subsequently merged using sequential I/Os.
To improve efficiency and minimize blocking, flushes and merges are often performed asynchronously in the background.

Despite their popularity,
LSM-trees have been criticized for suffering from write stalls and large performance variances~\cite{compaction-stall2017, blsm2012, ldc2019}.
To illustrate this problem, we conducted a micro-experiment on RocksDB~\cite{rocksdb},
a state-of-the-art LSM-based key-value store, to evaluate its write throughput on SSDs using the YCSB benchmark~\cite{ycsb2010}.
The instantaneous write throughput over time is depicted in \reffigure{fig:rocksdb-stall}.
As one can see, the write throughput of RocksDB periodically slows down after the first 300 seconds, which is when
the system has to wait for background merges to catch up.
Write stalls can significantly impact percentile write latencies and
must be minimized to improve the end-user experience or to meet strict service-level agreements~\cite{perf-predictability2018}.

\begin{figure}
	\centering
	\includegraphics[width=0.8\linewidth]{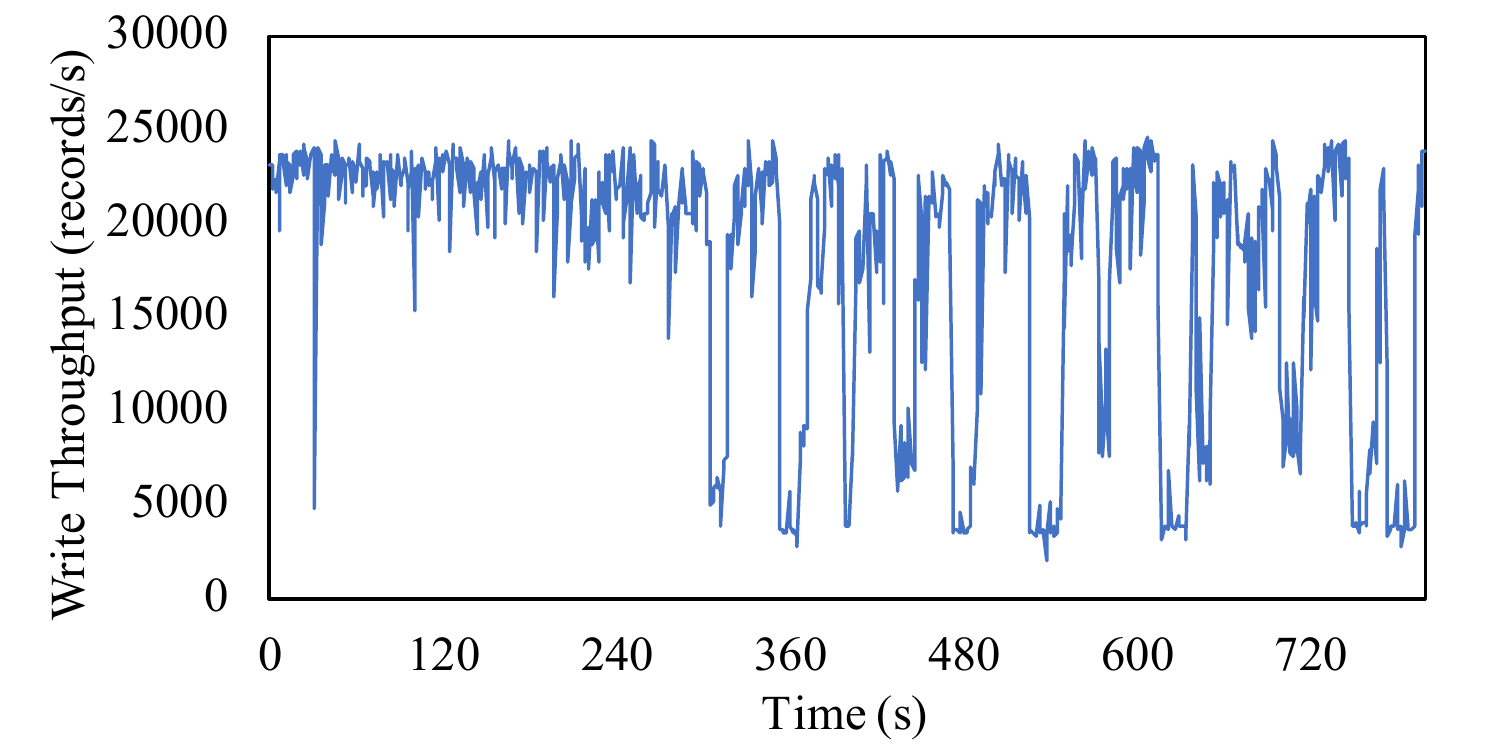}
	\vspace{-0.1in}
	\caption{Instantaneous write throughput of RocksDB: writes are periodically stalled to wait for lagging merges}
	\label{fig:rocksdb-stall}
\end{figure}

In this paper, we study the impact of write stalls and how to minimize write stalls for various LSM-tree designs.
It should first be noted that some write stalls are inevitable.
Due to the inherent mismatch between fast in-memory writes and slower background I/O operations,
in-memory writes must be slowed down or stopped if background flushes or merges cannot catch up.
Without such a flow control mechanism, the system will eventually run out of memory (due to slow flushes) or disk space (due to slow merges).
Thus, it is not a surprise that an LSM-tree can exhibit large write stalls if one measures its maximum write throughput
by writing as quickly data as possible, such as we did in \reffigure{fig:rocksdb-stall}.

This inevitability of write stalls does not necessarily limit the applicability of
LSM-trees since in practice writes do not arrive as quickly as possible, but rather are controlled by the expected data arrival rate.
The data arrival rate directly impacts the write stall behavior and resulting write latencies of an LSM-tree.
If the data arrival rate is relatively low, then write stalls are unlikely to happen.
However, it is also desirable to maximize the supported data arrival rate so that the system's resources can be fully utilized.
Moreover, the expected data arrival rate is subject to an important constraint - it must be smaller than the processing capacity of the target LSM-tree.
Otherwise, the LSM-tree will never be able to process writes as they arrive, causing infinite write latencies.
Thus, to evaluate the write stalls of an LSM-tree, the first step is to choose a proper data arrival rate.

As the first contribution, we propose a simple yet effective approach to evaluate the write stalls of various LSM-tree designs
by answering the following question: If we set the data arrival rate close to (e.g., 95\% of) the maximum write throughput of an LSM-tree,
will that cause write stalls?
In other words, can a given LSM-tree design provide both a high write throughput and a low write latency?
Briefly, the proposed approach consists of two phases: a \emph{testing} phase and a \emph{running} phase.
During the testing phase, we experimentally measure the maximum write throughput of an LSM-tree by simply writing as much data as possible.
During the running phase, we then set the data arrival rate close to the measured maximum write throughput as the limiting data arrival rate to evaluate its write stall behavior based on write latencies.
If write stalls happen, the measured write throughput is not \emph{sustainable} since it cannot be used in the long-term due to the large latencies.
However, if write stalls do not happen, then write stalls are no longer a problem since
the given LSM-tree can provide a high write throughput with small performance variance.

Although this approach seems to be straightforward at first glance, there exist two challenges that must be addressed.
First, how can we accurately measure the maximum sustainable write rate of an LSM-tree experimentally?
Second, how can we best schedule LSM I/O operations so as to minimize write stalls at runtime?
In the remainder of this paper, we will see that the merge scheduler of an LSM-tree can have a large impact on write stalls.
As the second contribution, we identify and explore the design choices for LSM merge schedulers
and present a new merge scheduler to address these two challenges.

As the paper's final contribution, we have implemented the proposed techniques and various LSM-tree designs inside Apache AsterixDB~\cite{asterixdb2014}.
This enabled us to carry out extensive experiments to evaluate the write stalls of LSM-trees and the effectiveness of the proposed techniques
using our two-phase evaluation approach.
We argue that with proper tuning and configuration,
LSM-trees can achieve both a high write throughput and small performance variance.

\shortonly{
The remainder of this paper is organized as follows:
\refsection{sec:background} provides background on LSM-trees and discusses related work.
\refsection{sec:experimental-setup} describes the general experimental setup used throughout this paper.
\refsection{sec:lsm-merge-scheduler} identifies the scheduling choices for LSM-trees
and experimentally evaluates bLSM's merge scheduler~\cite{blsm2012}.
Sections \ref{sec:full-merge} and \ref{sec:partitioning} present our techniques for minimizing write stalls for various LSM-tree designs.
\refsection{sec:summary} summarize the important lessons and insights from our evaluation.
Finally, \refsection{sec:conclusion} concludes the paper.
An extended version of this paper~\cite{lsm-stability-extend} further extends our evaluation to the size-tiered merge policy used in practical systems
and LSM-based secondary indexes.
}

\longonly{
The remainder of this paper is organized as follows:
\refsection{sec:background} provides background information on LSM-trees and briefly surveys related work.
\refsection{sec:experimental-setup} describes the general experimental setup used throughout this paper.
\refsection{sec:lsm-merge-scheduler} identifies the design choices for LSM merge schedulers
and evaluates bLSM's spring-and-gear scheduler~\cite{blsm2012}.
Sections \ref{sec:full-merge} and \ref{sec:partitioning} present our techniques for minimizing write stalls 
for full merges and partitioned merges respectively.
\refsection{sec:summary} summarizes the lessons and insights from our evaluation.
Finally, \refsection{sec:conclusion} concludes the paper.
}

\section{Background}
\label{sec:background}
\subsection{Log-Structured Merge Trees}
\label{sec:background-lsm}
The LSM-tree~\cite{lsm1996} is a persistent index structure optimized for write-intensive workloads.
In an LSM-tree, writes are first buffered into a memory component.
An insert or update simply adds a new entry with the same key,
while a delete adds an anti-matter entry indicating that a key has been deleted.
When the memory component is full, it is flushed to disk to form a new disk component, within which entries are ordered by keys.
Once flushed, LSM disk components are immutable.

A query over an LSM-tree has to reconcile the entries with identical keys from multiple components, as entries from newer components override those from older components.
A point lookup query simply searches all components from newest to oldest until the first match is found.
A range query searches all components simultaneously using a priority queue to perform reconciliation.
To speed up point lookups, a common optimization is to build Bloom filters~\cite{bloom-filter1970} over the sets of keys stored in disk components.
If a Bloom filter reports that a key does not exist, then that disk component can be excluded from searching.
As disk components accumulate, query performance tends to degrade since more components must be examined.
To counter this, smaller disk components are gradually merged into larger ones.
This is implemented by scanning old disk components to create a new disk component with unique entries.
The decision of what disk components to merge is made by a pre-defined \emph{merge policy}, which is discussed below.

\textbf{Merge Policy.}
Two types of LSM merge policies are commonly used in practice~\cite{lsm-survey}, both of which organize components into ``levels''.
The leveling merge policy (\reffigure{fig:merge-policy}a) maintains one component per level,
and a component at Level $i+1$ will be $T$ times larger than that of Level $i$.
As a result, the component at Level $i$ will be merged multiple times with the component from Level $i-1$ until it fills up
and is then merged into Level $i+1$.
In contrast, the tiering merge policy (\reffigure{fig:merge-policy}b) maintains multiple components per level.
When a Level $i$ becomes full with $T$ components, these $T$ components are merged together into a new component at Level $i+1$.
In both merge policies, $T$ is called the \emph{size ratio}, as it controls the maximum capacity of each level.
We will refer to both of these merge policies as \emph{full merges} since components are merged entirely.

\begin{figure}
	\centering
	\includegraphics[width=\linewidth]{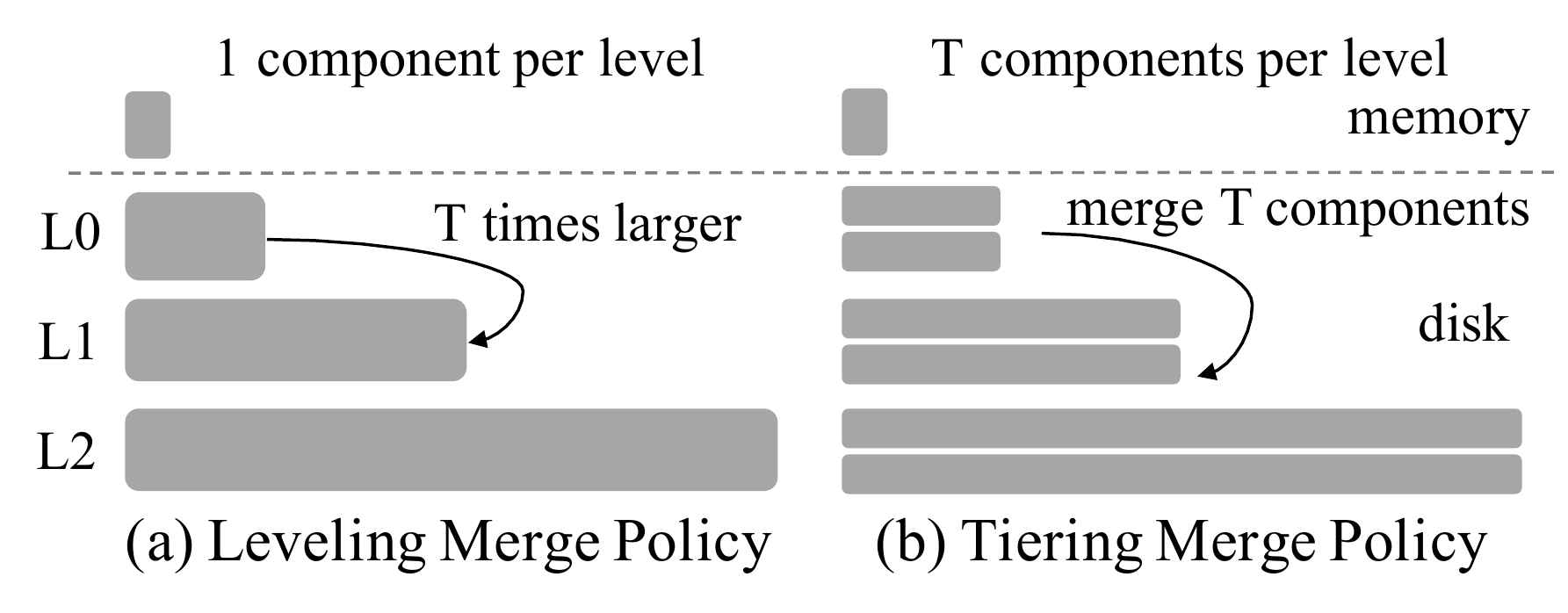}
	\vspace{-0.25in}
	\caption{LSM-tree Merge Policies}
	\label{fig:merge-policy}
\end{figure}

In general, the leveling merge policy optimizes for query performance by minimizing the number of components but at the cost of write performance.
This design also maximizes space efficiency, which measures the amount of space used for storing obsolete entries,
by having most of the entries at the largest level.
In contrast, the tiering merge policy is more write-optimized by reducing the merge frequency,
but this leads to lower query performance and space utilization.

\textbf{Partitioning.}
Partitioning is a commonly used optimization in modern LSM-based key-value stores
that is often implemented together with the leveling merge policy, as pioneered by LevelDB~\cite{leveldb}.
In this optimization, a large LSM disk component is range-partitioned into multiple (often fixed-size) {files}.
This bounds the processing time and the temporary space of each merge.
An example of a partitioned LSM-tree {with the leveling merge policy}
is shown in \reffigure{fig:partitioned-lsm}, where each {file} is labeled with its key range.
Note that partitioning starts from {Level} 1, as components in {Level} 0 are directly flushed from memory.
To merge a {file} from {Level} $i$ to {Level} $i+1$, all of its overlapping {files} at {Level} $i+1$ are selected and
these {files} are merged to form new {files} at {Level} $i+1$.
For example in \reffigure{fig:partitioned-lsm}, the {file} labeled 0-50 at {Level} 1 will be merged
with the files labeled 0-20 and 22-52 at {Level} 2,
which produce new files labeled 0-15, 17-30, and 32-52 at {Level} 2.
To select which file to merge next, LevelDB uses a round-robin policy.

\hyphenation{Wired-Tiger}

Both full merges and partitioned merges are widely used in existing systems.
Full merges are used in AsterixDB~\cite{asterixdb-web}, Cassandra~\cite{cassandra},  HBase~\cite{hbase}, ScyllaDB~\cite{sylladb},
Tarantool~\cite{tarantool}, and WiredTiger (MongoDB)~\cite{wiredtiger}.
Partitioned merges are used in LevelDB~\cite{leveldb}, RocksDB~\cite{rocksdb}, and X-Engine~\cite{xengine}.

\begin{figure}
	\centering
	\includegraphics[width=\linewidth]{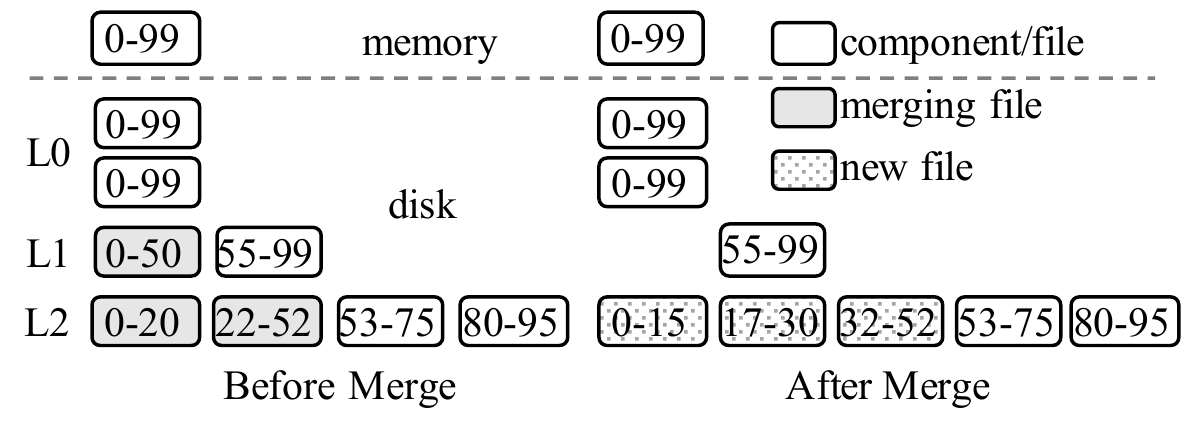}
	\vspace{-0.25in}
	\caption{{Partitioned LSM-tree with Leveling Merge Policy}}
	\label{fig:partitioned-lsm}
\end{figure}

\textbf{Write Stalls in LSM-trees.}
{Since in-memory writes are inherently faster than background I/O operations,
writing to memory components sometimes must be stalled to ensure the stability of an LSM-tree,
which, however, will negatively impact write latencies.
This is often referred to as the \emph{write stall} problem.}
If the incoming write speed is faster than the flush speed, writes will be stalled when all memory components are full.
Similarly, if there are too many disk components, new writes should be stalled as well.
In general, merges are the major source of write stalls since writes are flushed once but merged multiple times.
Moreover, flush stalls can be avoided by giving higher I/O priority to flushes.
In this paper, we thus focus on write stalls caused by merges.

\subsection{Apache AsterixDB}
\label{sec:asterixdb}
Apache AsterixDB~\cite{asterixdb-web, asterixdb2014,asterixdb2019} is a parallel, semi-structured Big Data Management System (BDMS)
that aims to manage large amounts of data efficiently.
It supports a feed-based framework for efficient data ingestion~\cite{asterixdb-feed2015,idea2019}.
The records of a dataset in AsterixDB are hash-partitioned based on their primary keys across multiple nodes of a shared-nothing cluster;
{thus, a range query has to search all nodes.}
Each partition of a dataset uses a primary LSM-based \btree index to store the data records,
while local secondary indexes, including LSM-based \btrees, R-trees, and inverted indexes, can be built to expedite query processing.
{AsterixDB internally uses a variation of the tiering merge policy to manage disk components, similar to the one used in existing systems~\cite{hbase,rocksdb}.
Instead of organizing components into levels explicitly as in \reffigure{fig:merge-policy}b, AsterixDB's variation simply
schedules merges based on the sizes of disk components.
In this work, we do not focus on the LSM-tree implementation in AsterixDB
but instead use AsterixDB as a common testbed to evaluate various LSM-tree designs.}

\subsection{Related Work}
\textbf{LSM-trees.}
Recently, a large number of improvements of the original LSM-tree~\cite{lsm1996} have been proposed.
\cite{lsm-survey} surveys these improvements, {ranging} from improving write performance~\cite{triad2017,dostoevsky2018,lsm-bush,hashkv2019,wisckey2017,sifrdb2018,pebblesdb2017,lwc-tree2017},
optimizing memory management~\cite{compaction2015,flodb2017,lsbm2017,elastic-bf2018},
supporting automatic tuning of LSM-trees~\cite{monkey2017,monkey-tods,lsm-model2016},
optimizing LSM-based secondary indexes~\cite{lsm-storage2019,secondary2018},
to extending the applicability of LSM-trees~\cite{umzi2019,slimdb2017}.
However, all of these efforts have largely ignored performance variances and write stalls of LSM-trees.

{Several LSM-tree implementations seek to bound the write processing latency
to alleviate the negative impact of write stalls~\cite{leveldb, rocksdb, fd-tree2010,ldc2019}.}
bLSM~\cite{blsm2012} proposes a spring-and-gear merge scheduler to avoid write stalls.
As shown in \reffigure{fig:blsm}, bLSM has one memory component, $C_0$, and two disk components, $C_1$ and $C_2$.
The memory component $C_0$ is continuously flushed and merged with $C_1$.
When $C_1$ becomes full, a new $C_1$ component is created while the old $C_1$, which now becomes $C'_1$, will be merged with $C_2$.
bLSM ensures that for each {Level} $i$, the progress of merging $C'_i$ into $C_{i+1}$ (denoted as ``$out_i$'')
will be roughly identical to the progress of the formation of a new $C_i$ (denoted as ``$in_i$'').
This eventually limits the write rate for the memory component ($in_0$) and avoids blocking writes.
{ However, we will see later that simply bounding the maximum write processing latency alone is insufficient, because a large variance in the processing rate can still cause large queuing delays for subsequent writes.
}

\begin{figure}
	\centering
	\includegraphics[width=1\linewidth]{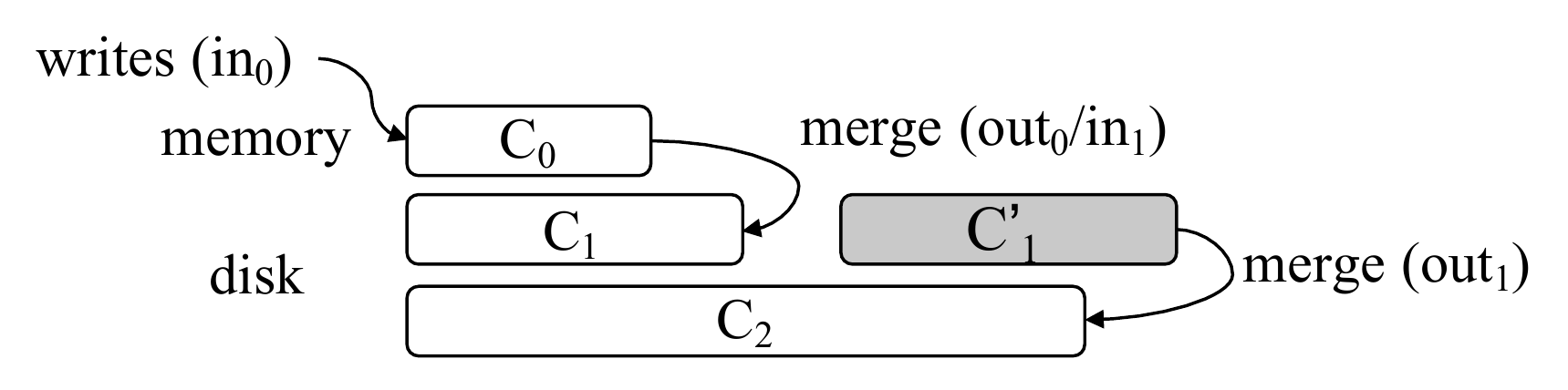}
	\vspace{-0.2in}
	\caption{bLSM's Spring-and-Gear Merge Scheduler}
	\label{fig:blsm}
\end{figure}

\textbf{Performance Stability.}
Performance stability has long been recognized as a critical performance metric.
The TPC-C benchmark~\cite{tpcc} measures not only absolute throughput,
but also specifies the acceptable upper bounds for the percentile latencies.
Huang et al.~\cite{perf-predictability2018} applied VProfiler~\cite{semantic-profiling2017}
to identify major sources of variance in database transactions.
Various techniques have been proposed to optimize the variance of query processing~\cite{piql2011,scale-indep2013,join-predictable2009,param-query-var2010,blink2009, perf-predictable2009}.
Cao et al.~\cite{storage-var2017} found that variance is common in storage stacks and heavily depends on configurations and workloads.
Dean and Barroso~\cite{google-latency} discussed several engineering techniques to
reduce performance variances at Google.
Different from these efforts, in this work we focus on the performance variances of LSM-trees
due to their inherent out-of-place update design.

\section{Experimental Methodology}
\label{sec:experimental-setup}
For ease of presentation, we will mix our techniques with a detailed performance analysis for each LSM-tree design.
We now describe the general experimental setup and methodology for all experiments to follow.

\subsection{Experimental Setup}
All experiments were run on a single node with an 8-core Intel i7-7567U 3.5GHZ CPU,
16 GB of memory, a 500GB SSD, and a 1TB 7200 rpm hard disk.
We used the SSD for LSM storage and configured the hard disk for transaction logging due to its sufficiently high sequential throughput.
We allocated 10GB of memory for the AsterixDB instance.
Within that allocation, the {buffer cache} size was set at 2GB.
Each LSM memory component had a 128MB budget, and each LSM-tree had two memory components to minimize stalls during flushes.
Each disk component had a Bloom filter with a false positive rate setting of 1\%.
The data page size was set at 4KB to {align with the SSD page size}.

It is important to note that not all sources of performance variance can be eliminated~\cite{perf-predictability2018}.
For example, writing a key-value pair with a 1MB value inherently requires more work than writing one that only has 1KB.
Moreover, short time periods with quickly occurring writes (workload bursts)
will be much more likely to cause write stalls than a long period of slow writes,
even though their long-term write rate may be the same.
{In this paper, we will focus on the \textit{avoidable variance}~\cite{perf-predictability2018}
caused by the internal implementation of LSM-trees instead of variances in the workloads.}

To evaluate the internal variances of LSM-trees,
we adopt YCSB~\cite{ycsb2010} as the basis for our experimental workload.
{Instead of using the pre-defined YCSB workloads, we designed our own workloads to better study the performance stability of LSM-trees.
Each experiment first loads an LSM-tree with 100 million records, in random key order, where each record has size 1KB.
It then runs for 2 hours to update the previously loaded LSM-tree.
This ensures that the measured write throughput of an LSM-tree is stable over time.
Unless otherwise noted, we used one writer thread for writing data to the LSM memory components.
We evaluated two update workloads, where the updated keys follow either a uniform or Zipf distribution.
The specific workload setups will be discussed in the subsequent sections.}

We used two commonly used I/O optimizations when implementing LSM-trees,
namely I/O throttling and periodic disk forces.
In all experiments, we throttled the {SSD} write speed of all LSM flush and merge operations to 100MB/s.
{This was implemented by using a rate limiter to inject artificial sleeps into SSD writes.
This mechanism bounds the negative impact of the SSD writes on query performance
and allows us to more fairly compare the performance differences of various LSM merge schedulers}.
We further had each flush or merge operation force its SSD writes after each 16MB of data.
This helps to limit the OS I/O queue length, {reducing the negative impact of SSD writes on queries.}
{We have verified that disabling this optimization would not impact the performance trends of writes;
however, large forces at the end of each flush and merge operation, which are required for durability,
can significantly interfere with queries.}

\subsection{Performance Metrics}
To quantify the impact of write stalls, we will not only present the write throughput of LSM-trees but also their write latencies.
However, there are different models for measuring write latencies.
Throughout the paper, we will use \emph{arrival rate} to denote the rate at which writes are submitted by clients,
\emph{processing rate} to denote the rate at which writes can be processed by an LSM-tree,
and \emph{write throughput} to denote the number of writes processed by an LSM-tree per unit of time.
The difference between the write throughput and arrival/processing rates is discussed further below.

The bLSM paper~\cite{blsm2012}, as well as most of the existing LSM research, used the experimental setup depicted in \reffigure{fig:queuing}a
to write as much data as possible and measure the latency of each write.
In this \emph{closed system} setup~\cite{queuing2013},
the processing rate essentially controls the arrival rate, which further equals the write throughput.
Although this model is sufficient for measuring the maximum write throughput of LSM-trees,
it is not suitable for characterizing their write latencies for several reasons.
First, writing to memory is inherently faster than background I/Os,
so an LSM-tree will always have to stall writes in order to wait for lagged flushes and merges.
Moreover, under this model, a client cannot submit its next write until its current write is completed.
Thus, when the LSM-tree is stalled, only a small number of ongoing writes will actually experience a large latency
since the remaining writes have not been submitted yet\footnote{The original release of the YCSB benchmark~\cite{ycsb2010}
	mistakenly used this model; this was corrected later in 2015~\cite{ycsb-bug}.}.

In practice, a DBMS generally cannot control how quickly writes are submitted by external clients,
nor will their writes always arrive as fast as possible.
Instead, the arrival rate is usually independent from the processing rate,
and when the system is not able to process writes as fast as they arrive, the newly arriving writes must be temporarily queued.
In such an \emph{open system} (\reffigure{fig:queuing}b), the measured write latency includes both the queuing latency and processing latency.
Moreover, an important constraint is that the arrival rate must be smaller than the processing rate since otherwise the queue length will be unbounded.
Thus, the (overall) write throughput is actually determined by the arrival rate.

\begin{figure}
	\centering
	\includegraphics[width=1\linewidth]{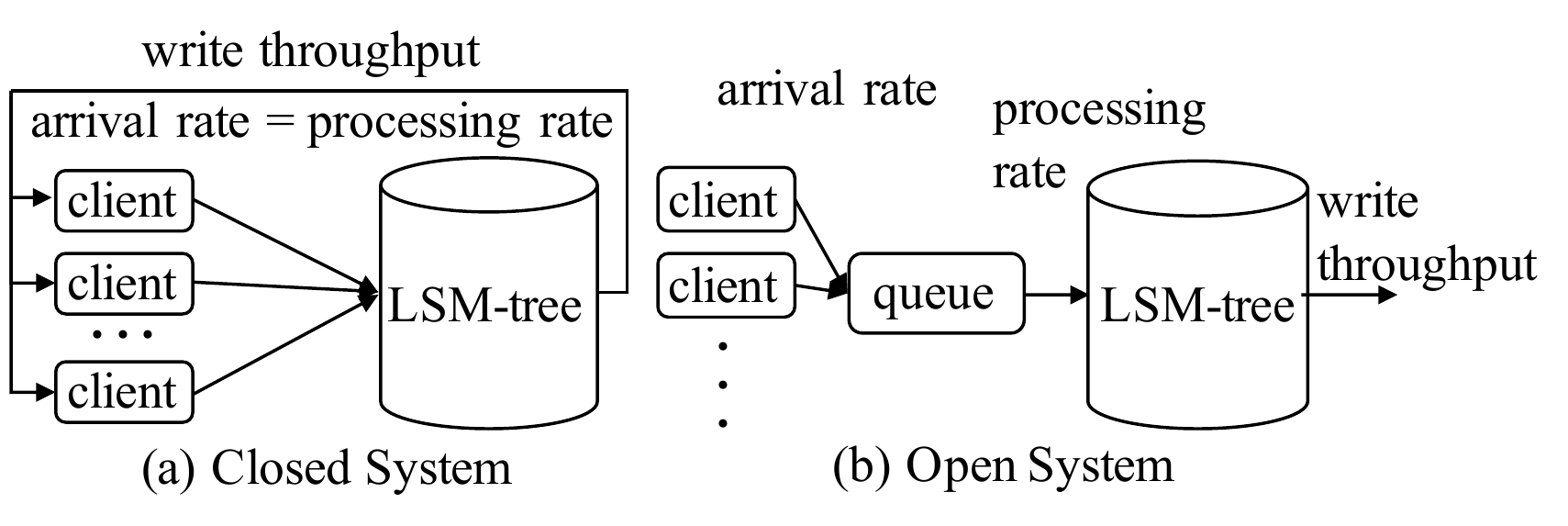}
	\vspace{-0.25in}
	\caption{Models for Measuring Write Latency}
	\label{fig:queuing}
\end{figure}

A simple example will illustrate the important difference between these two models.
Suppose that 5 clients are used to generate an intended arrival rate of 1000 writes/s
and that the LSM-tree stalls for 1 second.
Under the closed system model (\reffigure{fig:queuing}a), only 5 delayed writes will experience a write latency of 1s since the remaining
(intended) 995 writes simply will not occur.
However, under the open system model (\reffigure{fig:queuing}b), all 1000 writes will be queued and their average latency will be at least 0.5s.

To evaluate write latencies in an open system, one must first set the arrival rate properly
since the write latency heavily depends on the arrival rate.
It is also important to maximize the arrival rate to maximize the system's utilization.
For these reasons, we propose a two-phase evaluation approach with a \emph{testing} phase and a \emph{running} phase.
During the \emph{testing} phase, we use the closed system model (\reffigure{fig:queuing}a) to measure the maximum write throughput of an LSM-tree, which is also its processing rate.
When measuring the maximum write throughput,
we excluded the initial 20-minute period (out of 2 hours) of the testing phase
since {the initially loaded LSM-tree has a relatively small number of disk components at first}.
During the \emph{running} phase, we use the open system model (\reffigure{fig:queuing}b)
to evaluate the \emph{write latencies} under a constant arrival rate set at 95\% of the measured maximum write throughput.
Based on queuing theory~\cite{queuing2013}, the queuing time approaches infinity when the utilization,
which is the ratio between the arrival rate and the processing rate, approaches 100\%.
We thus empirically determine a high utilization load (95\%) while leaving some room for the system to absorb variance.
If the running phase then reports large write latencies,
the maximum write throughput as determined in the testing phase is not sustainable;
we must improve the implementation of the LSM-tree or reduce the expected arrival rate to reduce the latencies.
In contrast, if the measured write latency is small, then the given LSM-tree can provide
a high write throughput with a small performance variance.

\begin{figure*}
	\centering
	\begin{subfigure}[t]{.3\textwidth}
		\centering
		\includegraphics[width=\scale\linewidth]{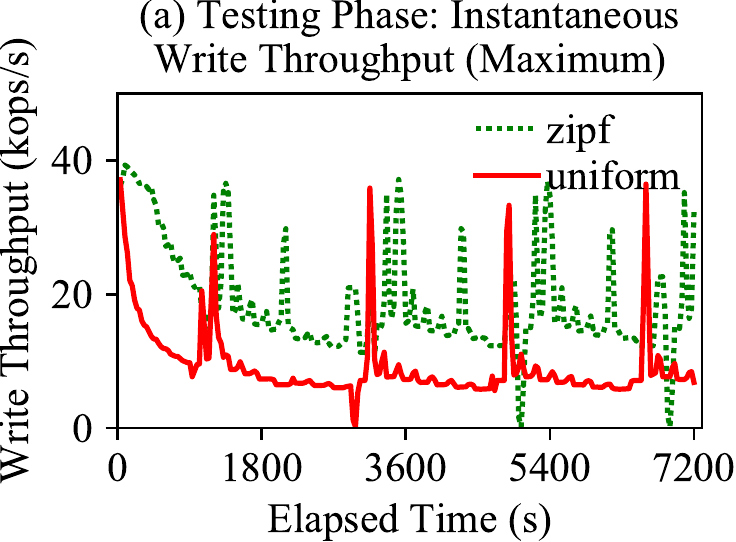}
	\end{subfigure}
	\hfil
	\begin{subfigure}[t]{.3\textwidth}
		\centering
		\includegraphics[width=\scale\linewidth]{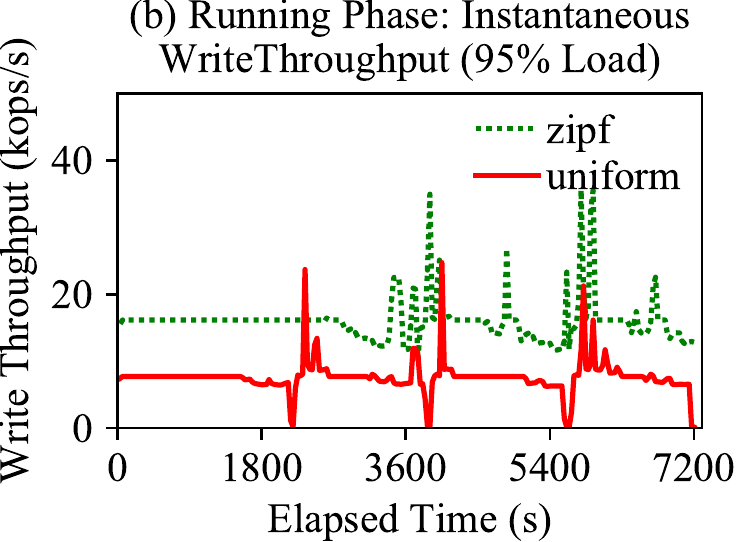}
	\end{subfigure}
	\hfil
	\begin{subfigure}[t]{.3\textwidth}
		\centering
		\includegraphics[width=\scale\linewidth]{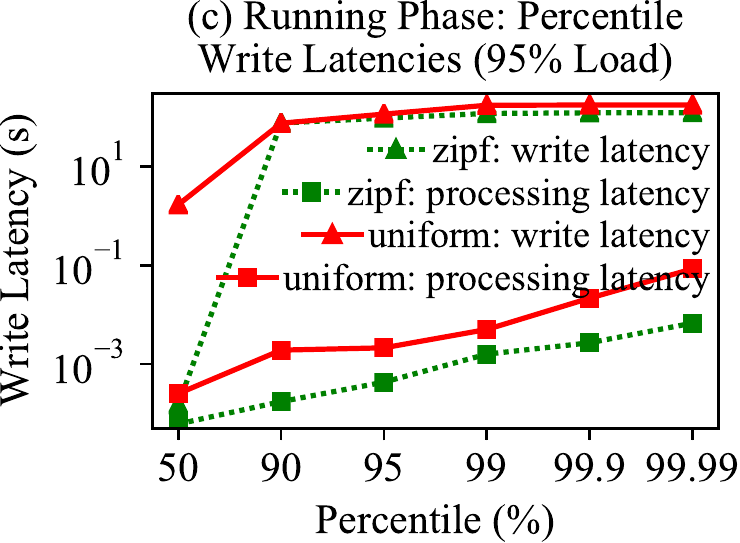}
	\end{subfigure}
	\vspace{-0.075in}
	\caption{Two-Phase Evaluation of bLSM}
	\label{fig:expr-blsm}
\end{figure*}

\section{LSM Merge Scheduler}
\label{sec:lsm-merge-scheduler}
Different from a merge policy, which decides which components to merge, {a \emph{merge scheduler} is responsible
for executing the merge operations created by the merge policy.}
In this section, we discuss the design choices for a merge scheduler and evaluate
bLSM's spring-and-gear merge scheduler.

\subsection{Scheduling Choices}
{The write cost of an LSM-tree, which is the number of I/Os per write, is determined by the LSM-tree design itself and the workload characteristics but not by how merges are executed~\cite{monkey2017}.
Thus, a merge scheduler will have little impact on the overall write throughput of an LSM-tree as long as the allocated I/O bandwidth budget can be fully utilized.}
However, different scheduling choices can significantly impact the write stalls of an LSM-tree,
and merge schedulers must be carefully designed to minimize write stalls.
We have identified the following design choices for a merge scheduler.

\textbf{Component Constraint:} A merge scheduler usually specifies
an upper-bound constraint on the total number of components allowed to accumulate before incoming writes to the LSM memory components should be stalled.
We call this the \emph{component constraint}.
For example, bLSM~\cite{blsm2012} allows at most two disk components per level,
while other systems like HBase~\cite{hbase} or Cassandra~\cite{cassandra} specify the total number of disk components across all levels.

\textbf{Interaction with Writes:} There exist different strategies to enforce a given component constraint.
One strategy is to {simply} stop processing writes once the component constraint is violated.
Alternatively, the processing of writes can be degraded gracefully based on the merge pressure~\cite{blsm2012}.

\textbf{Degree of Concurrency:} In general, an LSM-tree can often create multiple merge operations in the same time period.
A merge scheduler should decide how these merge operations should be scheduled.
Allowing concurrent merges will enable merges at multiple levels to proceed concurrently,
but they will also compete for CPU and I/O resources, which can negatively impact query performance~\cite{compaction2015}.
As two examples, bLSM~\cite{blsm2012} allows one merge operation per level,
while LevelDB~\cite{leveldb} uses just one single background thread to execute all merges one by one.

\textbf{I/O Bandwidth Allocation:} Given multiple concurrent merge operations, the merge scheduler should further
decide how to allocate the available {I/O bandwidth} among these merge operations. 
A commonly used heuristic is to allocate {I/O bandwidth} ``fairly" (evenly) to all active merge operations.
Alternatively, bLSM~\cite{blsm2012} allocates {I/O bandwidth} based on the relative progress of the merge operations to ensure that
merges at each level all make steady progress.

\subsection{Evaluation of bLSM}
\label{sec:blsm}
Due to the implementation complexity of bLSM and its dependency on a particular storage system, Stasis~\cite{stasis2006},
we chose to directly evaluate the released version of bLSM~\cite{blsmgithub}.
bLSM uses the leveling merge policy with two on-disk levels.
We set its memory component size to 1GB and size ratio to 10
so that the experimental dataset with 100 million records can fit into the last level.
We used 8 write threads to maximize the write throughput of bLSM.

\textbf{Testing Phase.}
During the testing phase, we measured the maximum write throughput of bLSM by writing as much data as possible
using both the uniform and Zipf update workloads.
The instantaneous write throughput of bLSM under these two workloads is shown in \reffigure{fig:expr-blsm}a.
For readability, the write throughput is averaged over 30-second windows.
(Unless otherwise noted, the same aggregation applies to all later experiments as well.)

{Even though bLSM's merge scheduler prevents writes from being stalled,
the instantaneous write throughput still exhibits a large variance with regular temporary peaks.
Recall that bLSM uses the merge progress at each level to control its in-memory write speed.
After the component $C_1$ is full and becomes $C'_1$, the original $C_1$ will be empty and will have much shorter merge times.
This will temporarily increase the in-memory write speed of bLSM, which then quickly drops as $C_1$ grows larger and larger.}
Moreover, the Zipf update workload increases the write throughput only because updated entries can be reclaimed earlier,
but the overall variance performance trends are still the same.

\textbf{Running Phase.}
Based on the maximum write throughput measured in the testing phase,
we then used a constant data arrival process (95\% of the maximum) in the running phase to evaluate bLSM's behavior.
\reffigure{fig:expr-blsm}b shows the instantaneous write throughput of bLSM under the uniform and Zipf update workloads.
bLSM maintains a sustained write throughput during the initial period of the experiment, but later has to slow down its in-memory write rate periodically due to
background merge pressure.
\reffigure{fig:expr-blsm}c further shows the resulting percentile write and processing latencies.
{The processing latency measures only the time for the LSM-tree to process a write,
while the write latency includes both the write's queuing time and processing time.}
By slowing down the in-memory write rate, bLSM indeed bounds the processing latency.
However, the write latency is much larger because writes must be queued when they cannot be processed immediately.
This suggests that simply bounding the maximum processing latency is far from sufficient;
it is important to minimize the variance in an LSM-tree's processing rate to minimize write latencies.

\section{Full Merges}
\label{sec:full-merge}
{In this section, we explore the scheduling choices of LSM-trees with full merges
	and then evaluate the impact of merge scheduling on write stalls using our two-phase approach.}
\shortlong{}{Finally, we examine other variations of the tiering merge policy that are used in practical systems.}

\subsection{Merge Scheduling for Full Merges}
We first introduce some useful notation for use throughout our analysis in \reftable{table:terms}.
To simplify the analysis, we will ignore the I/O cost of flushes since merges consume most of the I/O bandwidth.

\begin{table}
	\caption{List of notation used in this paper}
	\label{table:terms}
	\centering
	{
	\begin{tabular}{|c|l|l|}
	\hline
	\textbf{Term} & \textbf{Definition} & \textbf{Unit} \\
	\hline		$T$  & size ratio of the merge policy & \\
	\hline 		$L$  & the number of levels in an LSM-tree & \\
 	\hline		$M$ & memory component size & entries \\
	\hline		$B$ & I/O bandwidth & entries/s\\
	\hline		$\mu$ & write arrival rate & entries/s \\
	\hline		$W$ & write throughput of an LSM-tree & entries/s\\
	\hline
	\end{tabular}
}
\end{table}

\subsubsection{Component Constraint}
\label{sec:component-constraint}
To provide acceptable query performance and space utilization, the total number of disk components of an LSM-tree must be bounded.
We call this upper bound the \emph{component constraint}, and it can be enforced either \emph{locally} or \emph{globally}.
A local constraint specifies the maximum number of disk components per level.
For example, bLSM~\cite{blsm2012} uses a local constraint to allow at most two components per level.
A global constraint instead specifies the maximum number of disk components across all levels.
Here we argue that global component constraints will better minimize write stalls.
In addition to external factors, such as deletes or shifts in write patterns,
the merge time at each level inherently varies for leveling since the size of the component at
{Level} $i$ varies from 0 to $(T-1)\cdot M \cdot T^{i-1}$.
Because of this, bLSM cannot provide a high yet stable write throughput over time.
Global component constraints will better absorb this variance and minimize the write stalls .

It remains a question how to determine the maximum number of disk components for the component constraint.
In general, tolerating more disk components will increase the LSM-tree's ability to reduce write stalls and absorb write bursts,
but it will decrease query performance and space utilization.
Given the negative impact of stalls on write latencies, one solution is to 
tolerate a sufficient number of disk components to avoid write stalls while the worst-case query performance and space utilization are still bounded.
For example, one conservative constraint would be to tolerate twice the expected number of disk components,
e.g., $2 \cdot L$ components for leveling and $2\cdot T\cdot L$ components for tiering.

\subsubsection{Interaction with Writes}
\label{sec:full-merge-interact-writes}
When the component constraint is violated, the processing of writes by an LSM-tree has to be slowed down or stopped.
Existing LSM-tree implementations~\cite{leveldb, rocksdb, blsm2012} prefer to gracefully slow down the in-memory write rate
by adding delays to some writes.
This approach reduces the maximum processing latency, as large pauses are broken down into many smaller ones,
but the overall processing rate of an LSM-tree, which depends on the I/O cost of each write, is not affected.
Moreover, this approach will result in an even larger queuing latency.
There may be additional considerations for gracefully slowing down writes, but we argue that processing writes as quickly as possible minimizes the overall write latency, as stated by the following theorem.~\shortlong{See~\cite{lsm-stability-extend} for detailed proofs for all theorems.}{See the \refappendix{appendix:proof} for proofs for all theorems.}

\textsc{Theorem 1.} \textit{Given any data arrival process and any LSM-tree, processing writes as quickly as possible minimizes the latency of each write.}

\textsc{Proof Sketch.} Consider two merge schedulers $S$ and $S'$ which only differ in that $S$ may add arbitrary delays to writes while $S'$ processes writes as quickly as possible. For each write request $r$, $r$ must be completed by $S'$ no later than $S$ because the LSM-tree has the same processing rate but $S$ adds some delays to writes.

It should be noted that Theorem 1 only considers write latencies.
By processing writes as quickly as possible, disk components can stack up more quickly (up to the component constraint), which may negatively impact query performance.
Thus, a better approach may be to increase the write processing rate, e.g., by changing the structure of the LSM-tree. We leave the exploration of this direction as future work.

\subsubsection{Degree of Concurrency}
\label{sec:degree-concurrency}
{A merge policy can often create multiple merge operations simultaneously.
For full merges, we can show that a single-threaded scheduler that executes one merge at a time is not sufficient for minimizing write stalls.}
Consider a merge operation at {Level} $i$.
For leveling, the merge time varies from 0 to $\frac{M \cdot T^i}{B}$ because the size of the component at {Level} $i$ varies from
0 to $(T-1)\cdot M \cdot T^{i-1}$.
For tiering, each component has size $M\cdot T^{i-1}$ and merging $T$ components thus takes time $\frac{M \cdot T^i}{B}$.
Suppose the arrival rate is $\mu$.
Without concurrent merges, there would be $\frac{\mu}{M}\cdot \frac{M \cdot T^i}{B} = \frac{\mu \cdot T^i}{B}$
newly flushed components added while this merge operation is being executed, assuming that flushes can still proceed.

Our two-phase evaluation approach chooses the maximum write throughput of an LSM-tree as the arrival rate $\mu$.
For leveling, the maximum write throughput is approximately
$W_{level} = \frac{2\cdot B}{T \cdot L}$, as each entry is merged $\frac{T}{2}$ times per level.
For tiering, the maximum write throughput
is approximately $W_{tier} = \frac{B}{L}$, as each entry is merged only once per level.
By substituting $W_{level}$ and $W_{tier}$ for $\mu$,
one needs to tolerate at least $\frac{2 \cdot T^{i-1}}{L}$ flushed components for leveling
and $\frac{T^i}{L}$ flushed components for tiering to avoid write stalls.
{Since the term $T^i$ grows exponentially, a large number of flushed components will have to be tolerated when a large disk component is being merged.
Consider the leveling merge policy with a size ratio of $10$.
To merge a disk component at {Level} $5$, approximately $\frac{2\cdot 10^4}{5} = 4000$ flushed components
would need to be tolerated, which is highly unacceptable.}

{Clearly, concurrent merges must be performed to minimize write stalls.}
When a large merge is being processed, smaller merges can still be completed to reduce the number of components.
By the definition of the tiering and leveling merge policies, there can be at most one active merge operation per level.
Thus, given an LSM-tree with $L$ levels, at most $L$ merge operations can be scheduled concurrently.

\subsubsection{I/O Bandwidth Allocation}
\label{sec:bandwidth-allocation}
Given multiple active merge operations, the merge scheduler must {further} decide how to allocate {I/O bandwidth} to these operations.
A heuristic used by existing systems~\cite{cassandra,hbase,rocksdb} is to allocate {I/O bandwidth} fairly (evenly) to all ongoing merges.
We call this the \emph{fair} scheduler.
The fair scheduler ensures that all merges at different levels can proceed, thus eliminating potential starvation.
{Recall that write stalls occur when an LSM-tree has too many disk components, thus violating the component constraint.
It is unclear whether or not the fair scheduler can minimize write stalls by minimizing the number of disk components over time.
}

Recall that both the leveling and tiering merge policies always merge the same number of disk components at once.
We {propose a novel} \emph{greedy} scheduler that always allocates the full {I/O bandwidth}
to the merge operation with the smallest remaining number of bytes.
The greedy scheduler has a useful property that it minimizes the number of disk components over time for a given set of merge operations.

\textsc{Theorem 2.} \textit{Given any set of merge operations that process the same number of disk components and any {I/O bandwidth} budget,
the greedy scheduler minimizes the number of disk components at any time instant.}

{
\textsc{Proof Sketch.} Consider an arbitrary scheduler $S$ and the greedy scheduler $S'$. Given $N$ merge operations, we can show that $S'$ always completes the $i$-th ($1\le i\le N$) merge operation no later than $S$. This can be done by noting that $S'$ always processes the smallest merge operation first.
}

Theorem 2 only considers a set of statically created merge operations.
This conclusion may not hold in general because sometimes completing a large merge may enable the merge policy to create smaller merges,
which can then reduce the number of disk components more quickly.
Because of this, there actually exists no merge scheduler that
can always minimize the number of disk components over time, as stated by the following theorem.
However, as we will see in our later evaluation, the greedy scheduler is still a very effective heuristic for minimizing write stalls.

\textsc{Theorem 3.} \textit{Given any {I/O bandwidth} budget, no merge scheduler can minimize the number of disk components at any time instant for any data arrival process and any LSM-tree for a deterministic merge policy where all merge operations process the same number of disk components.}

{
\textsc{Proof Sketch.} Consider an LSM-tree that has created a small merge $M_S$ and a large merge $M_L$.
Completing $M_L$ allows the LSM-tree to create a new small merge $M'_S$ that is smaller than $M_S$.
Consider two merge schedulers $S_1$ and $S_2$, where $S_1$ first processes $M_S$ and then $M_L$,
and $S_2$ first processes $M_L$ and then $M'_S$.
It can be shown that $S_1$ has the earliest completion time for the first merge and $S_2$ has the earliest completion time for the second merge,
but no merge scheduler can outperform both $S_1$ and $S_2$.
}

\subsubsection{Putting Everything Together}
\label{sec:greedy-scheduler}
{Based on the discussion of each scheduling choice, we now summarize the proposed greedy scheduler.
The greedy scheduler enforces a global component constraint with a sufficient number of disk components, e.g., twice the expected number of components of an LSM-tree, to minimize write stalls while ensuring the stability of the LSM-tree.
It processes writes as quickly as possible and only stops the processing of writes when the component constraint is violated.
The greedy scheduler performs concurrent merges but allocates the full I/O bandwidth to the merge operation with the smallest remaining bytes. \shortonly{Whenever a merge operation is created or completed, the greedy scheduler is notified to find the smallest merge to execute next. Thus, a large merge can be interrupted by a newly created smaller merge. However, in general one cannot exactly know which merge operation requires the least amount of {I/O bandwidth} until the new component is fully produced. To handle this, the smallest merge operation can be approximated by using the number of remaining pages of the merging components.}}

\longonly{
The pseudocode for the greedy scheduling algorithm is shown in \reffigure{alg:greedy-scheduler}.
It stores the list of scheduled merge operations in \emph{mergeOps}.
At any time, there is at most one merge operation being executed by the merge thread, which is denoted by \emph{activeOp}.
The merge policy calls \textsc{ScheduleMerge} when a new merge operation is scheduled,
and the merge thread calls \textsc{CompleteMerge} when a merge operation is completed.
In both functions, \emph{mergeOps} is updated accordingly and the merge scheduler is notified to check
whether a new merge operation needs to be executed.
It should be noted that in general one cannot exactly know which merge operation requires the least amount of {I/O bandwidth}
until the new component has been fully produced.
Thus, line 12 uses the number of remaining input pages as an approximation to determine the smallest merge operation.
Finally, if the newly selected merge operation is inactive, i.e., not being executed,
the scheduler pauses the previous active merge operation and activates the new one.
}

\longonly{
\begin{figure}
	\centering
	\small
	\begin{algorithmic}[1]
		\State{mergeOps $\gets$ the list of scheduled merge operations}
		\State{activeOp $\gets$ the active merge operation}
		
		\Function{ScheduleMerge}{newOp}
			\State mergeOps.\Call{Add}{newOp}
			\State notify GreedyScheduler
		\EndFunction
		\Function{CompleteMerge}{}
			\State mergeOps.\Call{Remove}{activeOp}
			\State activeOp $\gets$ NULL
			\State notify GreedyScheduler
		\EndFunction
		\Function{GreedyScheduler}{}
			\While{mergeOps changes}
				\State newOp $\gets$ find the merge operation with the fewest remaining input pages in $mergeOps$
				\If{newOp $\ne$ NULL AND newOp $\ne$ activeOp}
					\State pause activeOp
					\State resume newOp
					\State activeOp $\gets$ newOp
				\EndIf
		    \EndWhile
		\EndFunction
	\end{algorithmic}
	\caption{Pseudocode for Greedy Scheduling Algorithm}
	\label{alg:greedy-scheduler}
\end{figure}
}

{Under the greedy scheduler, larger merges may be starved at times since they receive lower priority.
This has a few implications.
First, during normal user workloads, such starvation can only occur if the arrival rate is temporarily faster than the processing rate of an LSM-tree.
Given the negative impact of write stalls on write latencies, it can actually be beneficial to temporarily delay large merges
so that the system can better absorb write bursts.
Second, the greedy scheduler should not be used in the testing phase
because it would report a higher but unsustainable write throughput due to such starved large merges.
}

Finally, our discussions of the greedy scheduler as well as the single-threaded scheduler
are based on an important assumption that a single merge operation is able to fully utilize the {available I/O bandwidth} budget.
Otherwise, multiple merges must be executed at the same time.
It is straightforward to extend the greedy scheduler to execute the smallest $k$ merge operations,
where $k$ is the degree of concurrency needed to fully utilize the {I/O bandwidth} budget.

\subsection{Experimental Evaluation}
We now experimentally evaluate the write stalls of LSM-trees using our two-phase approach.
{We discuss the specific experimental setup followed by the detailed evaluation,
including the impact of merge schedulers on write stalls,
the benefit of enforcing the component constraint globally and of processing writes as quickly as possible,
and the impact of merge scheduling on query performance.}

\subsubsection{Experimental Setup}
{All experiments in this section were performed using AsterixDB
with the general setup described in \refsection{sec:experimental-setup}.}
Unless otherwise noted, the size ratio of leveling was set at 10,
which is a commonly used configuration in practice~\cite{leveldb,rocksdb}.
For the experimental dataset with 100 million unique records, this results in a three-level LSM-tree,
where the last level is nearly full.
{For tiering, the size ratio was set at 3, which leads to better write performance than leveling without sacrificing too much on query performance. This ratio results in an eight-level LSM-tree.}


\begin{figure}
	\centering
	\begin{minipage}[t]{.3\textwidth}
		\centering
		\includegraphics[width=\linewidth]{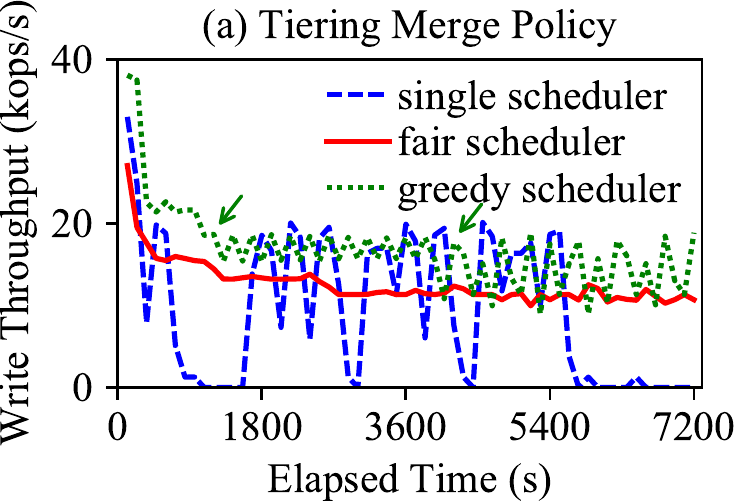}
		\vspace{-0.1in}
	\end{minipage}
	\hfil
	\begin{minipage}[t]{.3\textwidth}
		\centering
		\includegraphics[width=\linewidth]{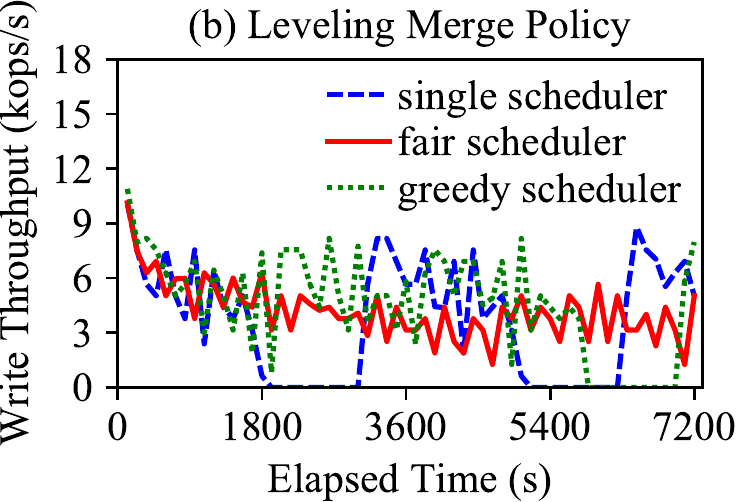}
	\end{minipage}
	\vspace{-0.05in}
	\caption{Testing Phase: Instantaneous Write Throughput}
	\label{fig:expr-max}
\end{figure}

\begin{figure*}[h]
	\centering
	\begin{minipage}[t]{.3\textwidth}
		\centering
		\includegraphics[width=\scale\linewidth]{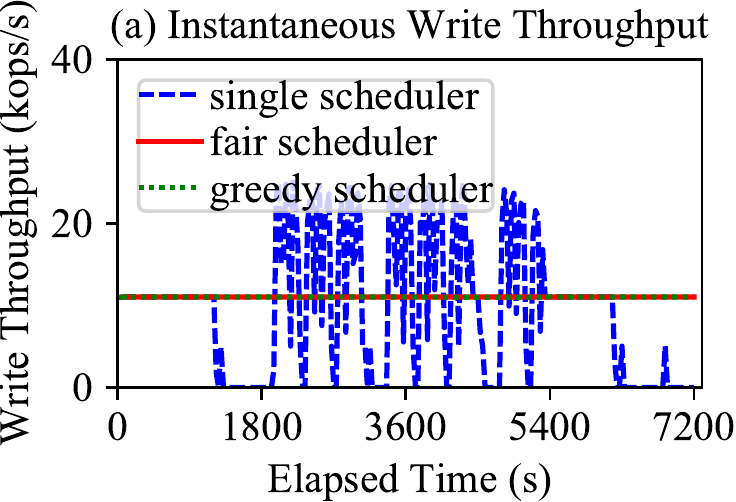}
	\end{minipage}
	\hfil
	\begin{minipage}[t]{.3\textwidth}
		\centering
		\includegraphics[width=\scale\linewidth]{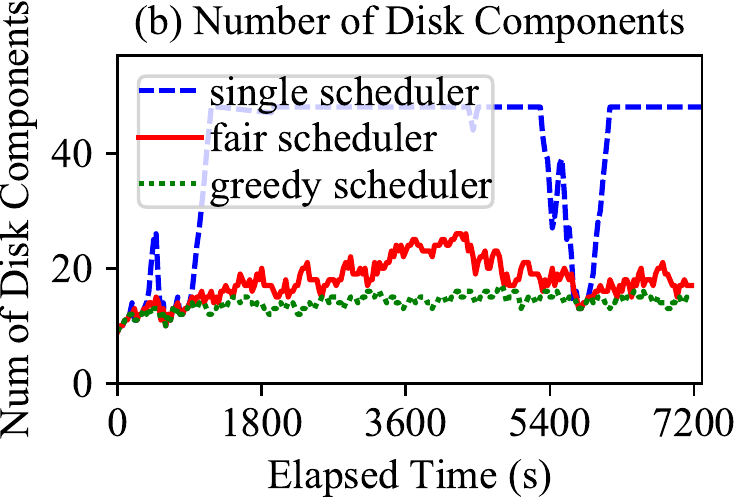}
	\end{minipage}
	\hfil
	\begin{minipage}[t]{.3\textwidth}
		\centering
		\includegraphics[width=\scale\linewidth]{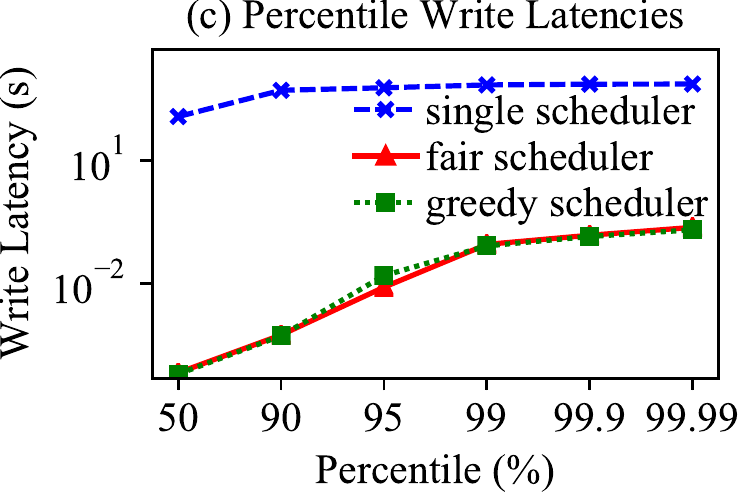}
	\end{minipage}%
	\vspace{-0.05in}
	\caption{Running Phase of Tiering Merge Policy (95\% Load)}
	\label{fig:expr-tier-open}
\end{figure*}

\begin{figure*}
	\centering
	\begin{minipage}[t]{.325\textwidth}
		\centering
		\includegraphics[width=\scale\linewidth]{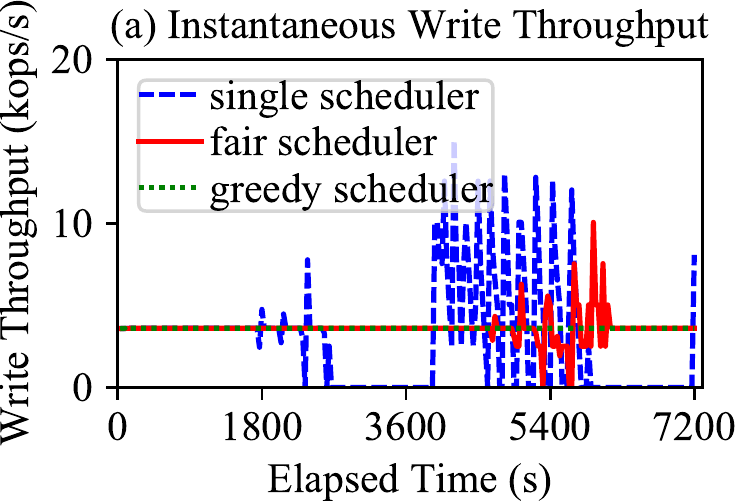}
	\end{minipage}
	\hfil
	\begin{minipage}[t]{.325\textwidth}
		\centering
		\includegraphics[width=\scale\linewidth]{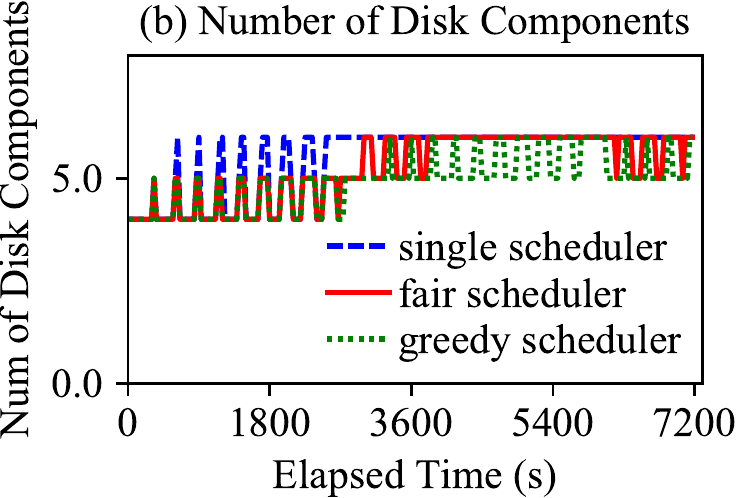}
	\end{minipage}
	\hfil
	\begin{minipage}[t]{.325\textwidth}
		\centering
		\includegraphics[width=\scale\linewidth]{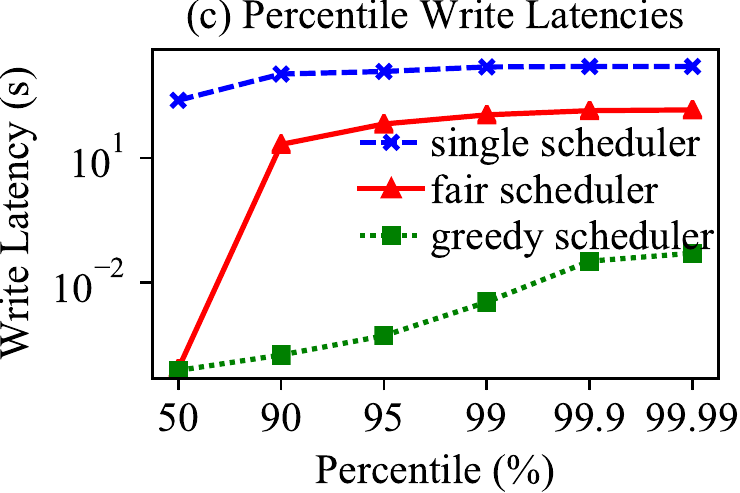}
	\end{minipage}
	\vspace{-0.05in}
	\caption{Running Phase of Leveling Merge Policy (95\% Load)}
	\label{fig:expr-level-open}
\end{figure*}

{
We evaluated the single-threaded scheduler (\refsection{sec:degree-concurrency}), the fair scheduler (\refsection{sec:degree-concurrency}), and the proposed greedy scheduler (\refsection{sec:greedy-scheduler}).
The single-threaded scheduler only executes one merge at a time using a single thread.
Both the fair and greedy schedulers are concurrent schedulers that execute each merge using a separate thread.
The difference is that the fair scheduler allocates the I/O bandwidth to all ongoing merges evenly,
while the greedy scheduler always allocates the full I/O bandwidth to the smallest merge.
To minimize flush stalls, a flush operation is always executed in a separate thread and receives higher I/O priority.
Unless otherwise noted, all three schedulers enforce global component constraints and process writes as quickly as possible.
The maximum number of disk components is set at twice the expected number of disk components for each merge policy.}
Each experiment was performed under both the uniform and Zipf update workloads.
Since the Zipf update workload had little impact on the overall performance trends, except that it led to higher write throughput, its experiment results are omitted here for brevity.

\subsubsection{Testing Phase}
During the testing phase, we measured the maximum write throughput of an LSM-tree by writing as much data as possible.
In general, alternative merge schedulers have little impact on the maximum write throughput since the {I/O bandwidth} budget is fixed,
but their measured write throughput may be different due to the finite experimental period.

Figures \ref{fig:expr-max}a and \ref{fig:expr-max}b
shows the instantaneous write throughout of LSM-trees using different merge schedulers for tiering and leveling.
Under both merge policies, the single-threaded scheduler regularly exhibits long pauses, making its write throughput vary over time.
The fair scheduler exhibits a relatively stable write throughput over time since all merge levels can proceed at the same rate.
With leveling, its write throughput still varies slightly over time since the component size at each level varies.
The greedy scheduler appears to achieve a higher write throughput than the fair scheduler by starving large merges.
However, this higher write throughput eventually drops when no small merges can be scheduled.
For example, the write throughput with tiering drops slightly at 1100s and 4000s,
and there is a long pause from 6000s to 7000s with leveling.
This result confirms that the fair scheduler is more suitable for testing the maximum write throughput of an LSM-tree,
as merges at all levels can proceed at the same rate.
In contrast, the single-threaded scheduler incurs many long pauses, causing a large variance in the measured write throughput.
The greedy scheduler provides a higher write throughput by starving large merges, which would be undesirable at runtime.

\subsubsection{Running Phase}
{Turning to the running phase, we used a constant data arrival process, configured based on 95\% of the maximum write throughput measured by the fair scheduler, to evaluate the write stalls of LSM-trees.}

\textbf{LSM-trees can provide a stable write throughput.}
We first evaluated whether LSM-trees with different merge schedulers can support a high write throughput with low write latencies.
For each experiment, we measured the instantaneous write throughput and the number of disk components over time
as well as percentile write latencies.

The results for tiering are shown in \reffigure{fig:expr-tier-open}.
Both the fair and greedy schedulers are able to provide stable write throughputs
and the total number of disk components never reaches the configured threshold.
The greedy scheduler also minimizes the number of disk components over time.
The single-threaded scheduler, however, causes a large number of write stalls due to the blocking of large merges,
which confirms our previous analysis.
Because of this, the single-threaded scheduler incurs large percentile write latencies.
In contrast, both the fair and greedy schedulers provide small write latencies because of their stable write throughput.
\reffigure{fig:expr-level-open} shows the corresponding results for leveling.
The single-threaded scheduler again performs poorly, causing a lot of stalls and thus large write latencies.
Due to the inherent variance of merge times, the fair scheduler alone cannot {provide} a stable write throughput;
this results in relatively large write latencies.
In contrast, the greedy scheduler avoids write stalls by always minimizing the number of components, which results in small write latencies.

This experiment confirms that LSM-trees can achieve a stable write throughput with a relatively small performance variance.
Moreover, the write stalls of an LSM-tree heavily depend on the design of the merge scheduler.

\textbf{Impact of Size Ratio.}
To verify our findings on LSM-trees with different shapes, we further carried out a set of experiments by varying the size ratio from 2 to 10 for
both tiering and leveling.
For leveling, we applied the dynamic level size optimization~\cite{rocksdb-space2017} so that the largest level remains almost full by slightly modifying the size ratio between {Levels} 0 and 1.
This optimization maximizes space utilization without impacting write or query performance.

During the testing phase, we measured the maximum write throughput for each LSM-tree configuration using the fair scheduler,
which is shown in \reffigure{fig:expr-size-ratio}a.
In general, a larger size ratio increases write throughput for tiering but decreases write throughput for leveling
because it decreases the merge frequency of tiering but increases that of leveling.
During the running phase, we evaluated the 99\% percentile write latency for each LSM-tree configuration using constant data arrivals,
which is shown in \reffigure{fig:expr-size-ratio}b.
With tiering, both the fair and greedy schedulers are able to provide a stable write throughput with small write latencies.
With leveling, the fair scheduler causes large write latencies when the size ratio becomes larger, as we have seen before.
In contrast, the greedy scheduler is always able to provide a stable write throughput along with small write latencies.
This again confirms that LSM-trees, despite their size ratios, can provide a high write throughput with a small variance
with an appropriately chosen merge scheduler.

\begin{figure}
	\centering
	\begin{subfigure}[t]{.227\textwidth}
		\centering
		\includegraphics[width=\linewidth]{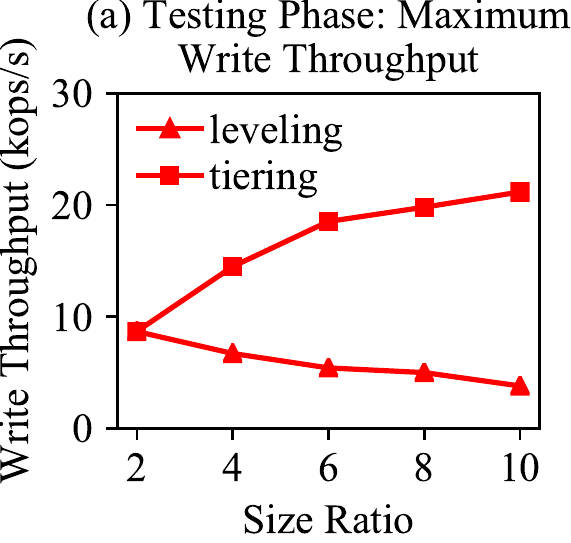}
	\end{subfigure}
	\hfil
	\begin{subfigure}[t]{.243\textwidth}
		\centering
		\includegraphics[width=\linewidth]{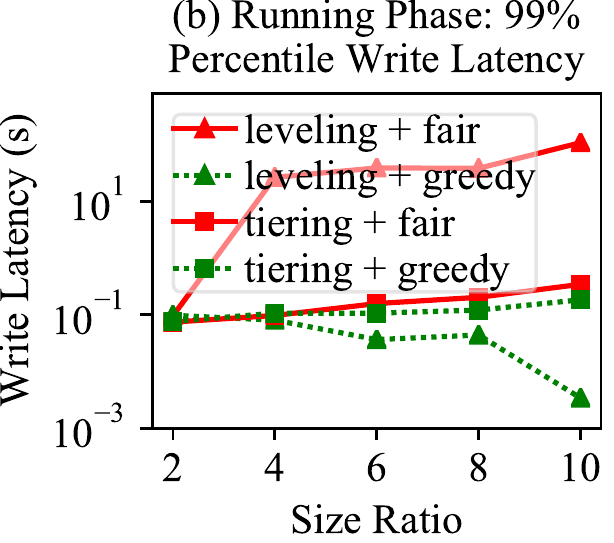}
	\end{subfigure}
	\vspace{-0.2in}
	\caption{Impact of Size Ratio on Write Stalls}
	\label{fig:expr-size-ratio}
\end{figure}

\begin{figure}
	\centering
	\begin{subfigure}[t]{.235\textwidth}
		\centering
		\includegraphics[width=\linewidth]{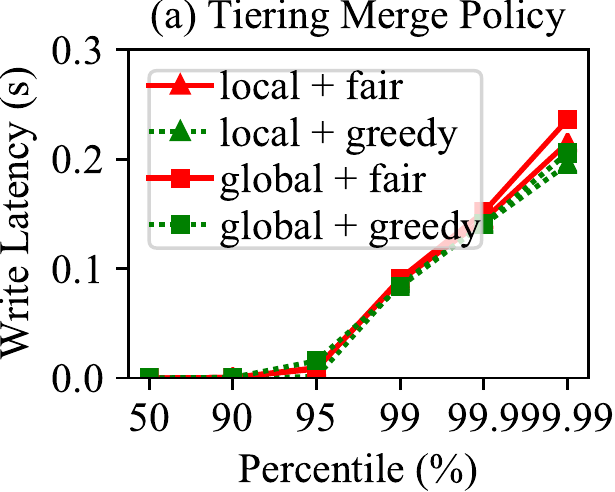}
	\end{subfigure}
	\hfil
	\begin{subfigure}[t]{.235\textwidth}
		\centering
		\includegraphics[width=\linewidth]{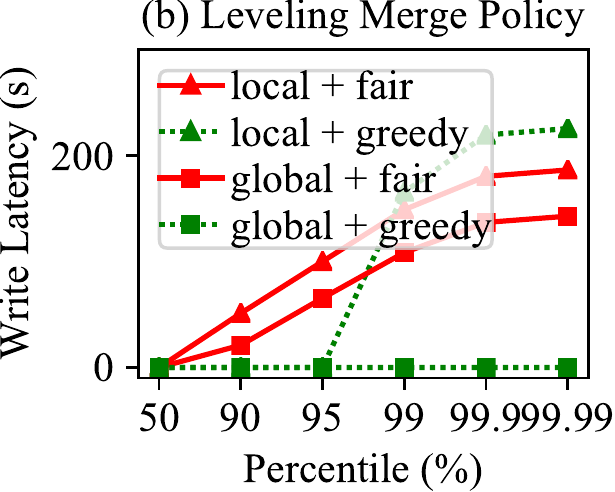}
	\end{subfigure}
	\vspace{-0.2in}
	\caption{{Impact of Enforcing Component Constraints on Percentile Write Latencies}}
	\label{fig:expr-component-constraint}
\end{figure}

\begin{figure}
	\centering
	\begin{subfigure}[t]{.235\textwidth}
		\centering
		\includegraphics[width=\linewidth]{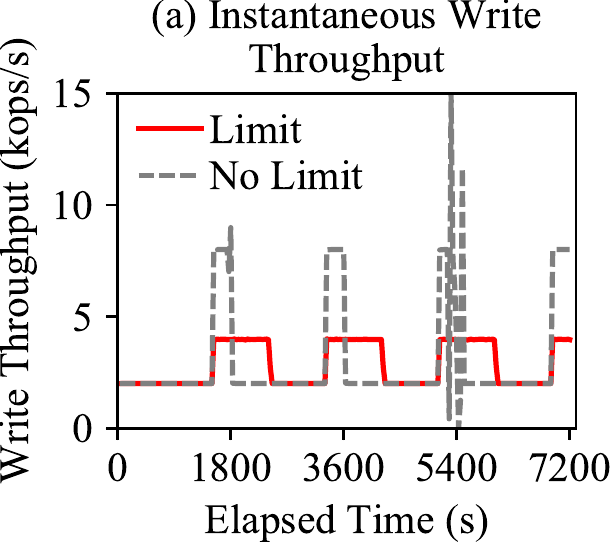}
	\end{subfigure}
	\hfil
	\begin{subfigure}[t]{.235\textwidth}
		\centering
		\includegraphics[width=\linewidth]{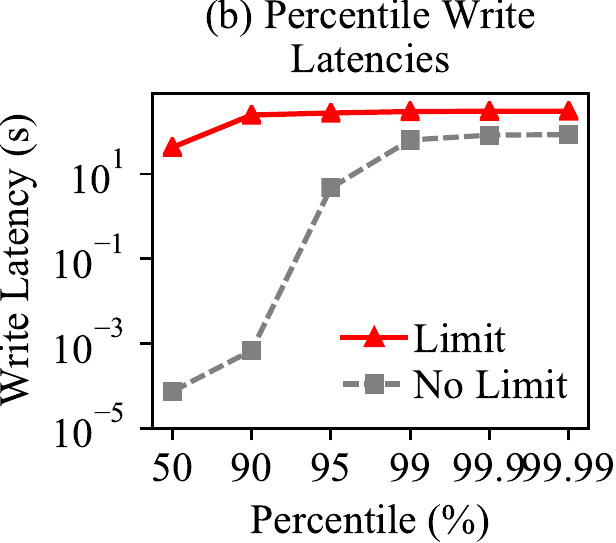}
	\end{subfigure}
	\vspace{-0.2in}
	\caption{Running Phase with Burst Data Arrivals}
	\label{fig:expr-level-burst}
\end{figure}

\textbf{Benefit of Global Component Constraints.}
{We next evaluated the benefit of global component constraints in terms of minimizing write stalls.
We additionally included a variation of the fair and greedy schedulers that enforces local component constraints, that is,
2 components per level for leveling and $2\cdot T$ components per level for tiering.}

The resulting write latencies are shown in \reffigure{fig:expr-component-constraint}.
In general, local component constraints have little impact on tiering since its merge time per level is relatively stable.
However, the resulting write latencies for leveling become much large due to the inherent variance of its merge times.
{Moreover, local component constraints have a larger negative impact on the greedy scheduler.
The greedy scheduler prefers small merges, which may not be able to complete due to possible violations of the constraint at the next level.
This in turn causes longer stalls and thus larger percentile write latencies.}
In contrast, global component constraints better absorb these variances, reducing the write latencies.

\shortonly{
\begin{figure*}[!h]
	\centering
	\begin{minipage}[t]{.3\textwidth}
		\centering
		\includegraphics[width=\scale\linewidth]{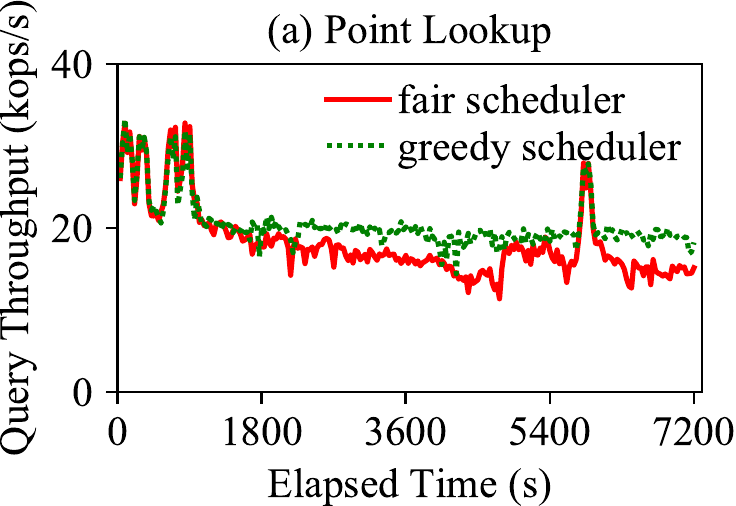}
	\end{minipage}
	\hfil
	\begin{minipage}[t]{.3\textwidth}
		\centering
		\includegraphics[width=\scale\linewidth]{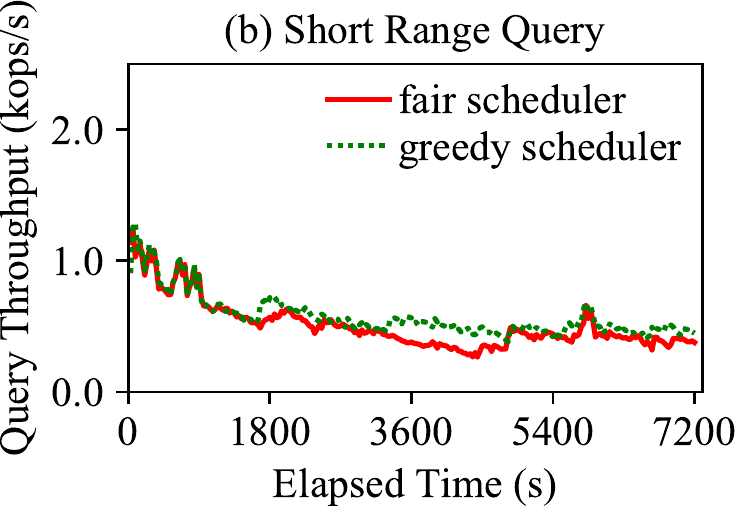}
	\end{minipage}
	\hfil
	\begin{minipage}[t]{.3\textwidth}
		\centering
		\includegraphics[width=\scale\linewidth]{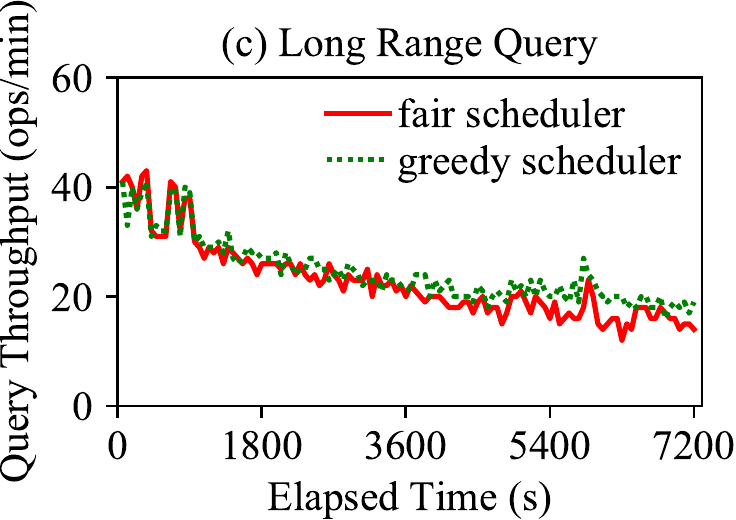}
	\end{minipage}
	\caption{Instantaneous Query Throughput of Tiering Merge Policy}
	\label{fig:expr-query-tier}
\end{figure*}

\begin{figure*}
	\centering
	\begin{minipage}[t]{.3\textwidth}
		\centering
		\includegraphics[width=\scale\linewidth]{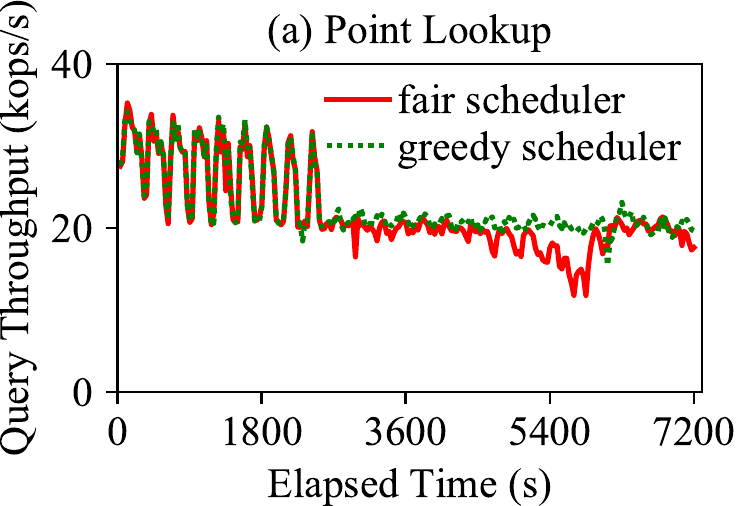}
	\end{minipage}
	\hfil
	\begin{minipage}[t]{.3\textwidth}
		\centering
		\includegraphics[width=\scale\linewidth]{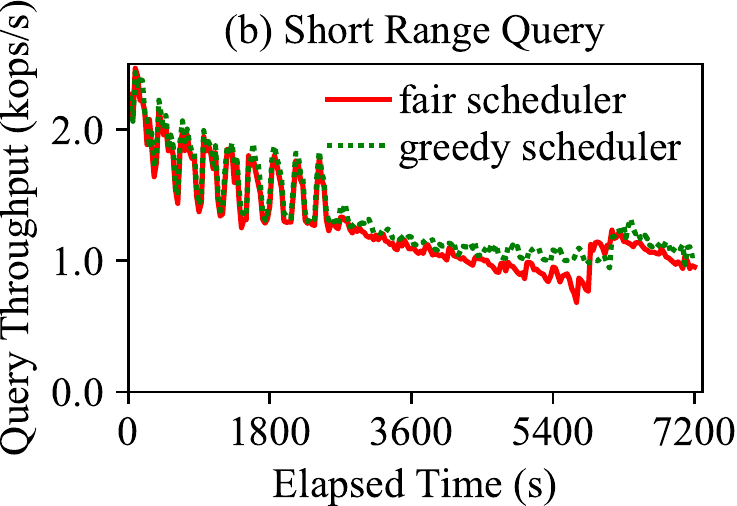}
	\end{minipage}
	\hfil
	\begin{minipage}[t]{.3\textwidth}
		\centering
		\includegraphics[width=\scale\linewidth]{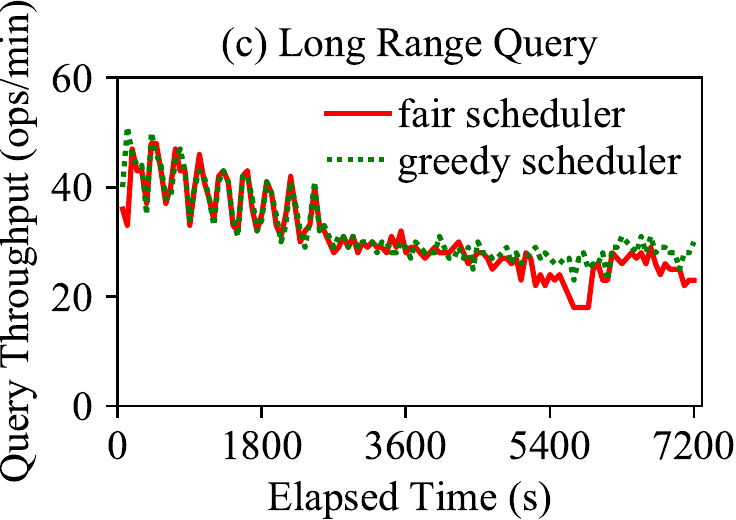}
	\end{minipage}
	\hfil
	\caption{Instantaneous Query Throughput of Leveling Merge Policy}
	\label{fig:expr-query-level}
\end{figure*}
}

\longonly{
	\begin{figure*}[!h]
		\centering
		\begin{minipage}[t]{.32\textwidth}
			\centering
			\includegraphics[width=\scale\linewidth]{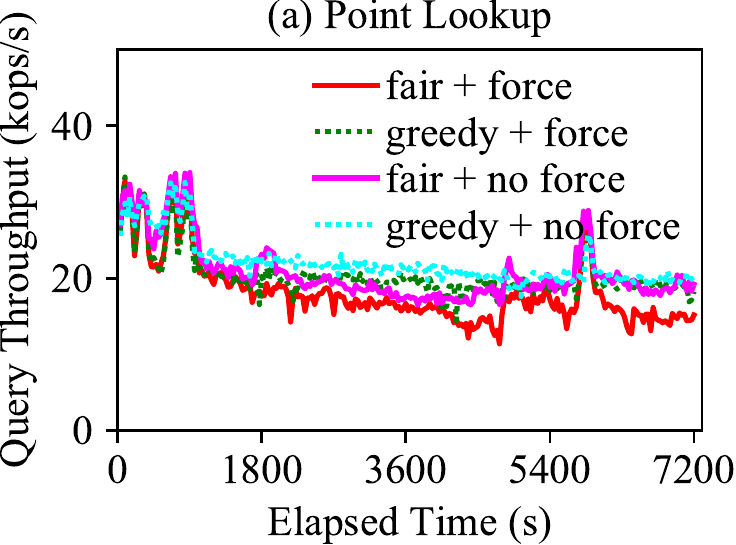}
		\end{minipage}
		\hfil
		\begin{minipage}[t]{.32\textwidth}
			\centering
			\includegraphics[width=\scale\linewidth]{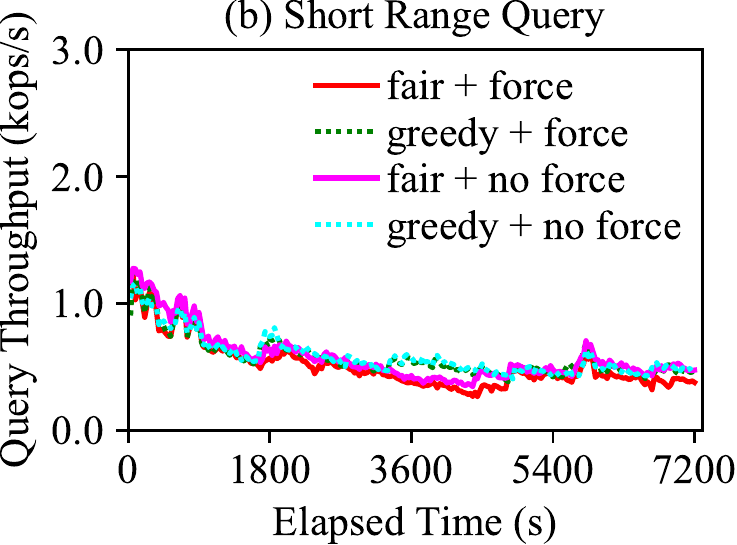}
		\end{minipage}
		\hfil
		\begin{minipage}[t]{.32\textwidth}
			\centering
			\includegraphics[width=\scale\linewidth]{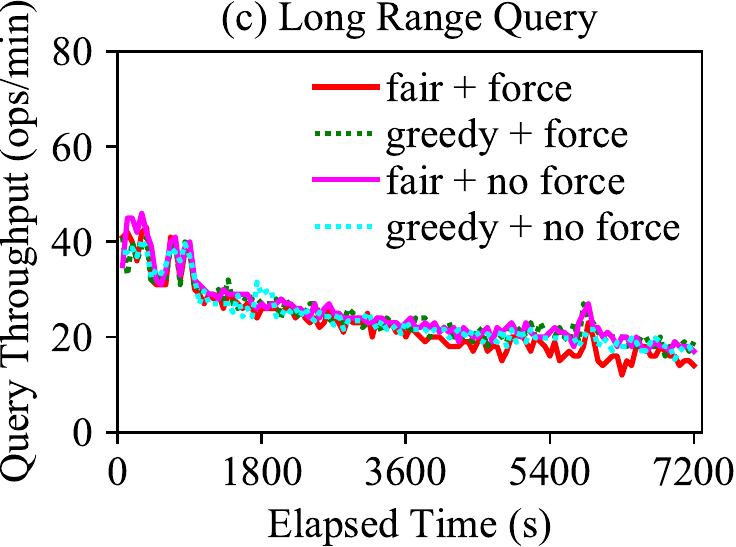}
		\end{minipage}
		\caption{Instantaneous Query Throughput of Tiering Merge Policy}
		\label{fig:expr-query-tier-extend}
	\end{figure*}
	\begin{figure*}
	\centering
	\begin{minipage}[t]{.32\textwidth}
		\centering
		\includegraphics[width=\scale\linewidth]{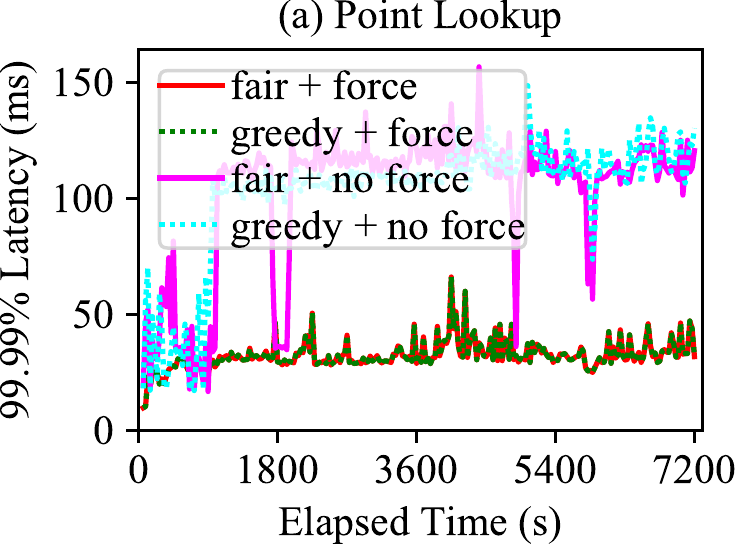}
	\end{minipage}
	\hfil
	\begin{minipage}[t]{.32\textwidth}
		\centering
		\includegraphics[width=\scale\linewidth]{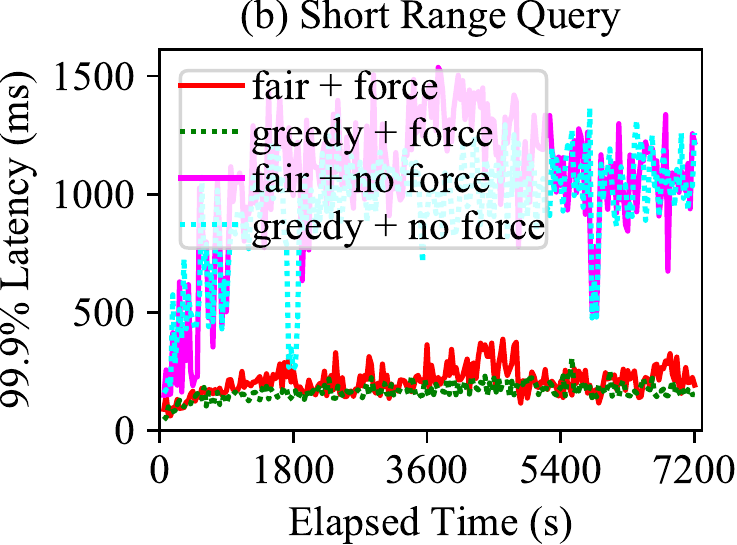}
	\end{minipage}
	\hfil
	\begin{minipage}[t]{.32\textwidth}
		\centering
		\includegraphics[width=\scale\linewidth]{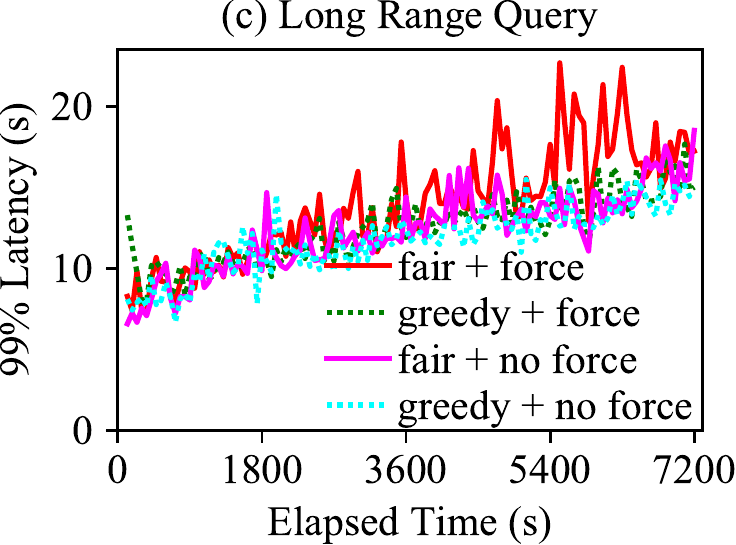}
	\end{minipage}
	\caption{Percentile Query Latencies of Tiering Merge Policy}
	\label{fig:expr-query-tier-latency}
\end{figure*}
	
	\begin{figure*}
		\centering
		\begin{minipage}[t]{.32\textwidth}
			\centering
			\includegraphics[width=\scale\linewidth]{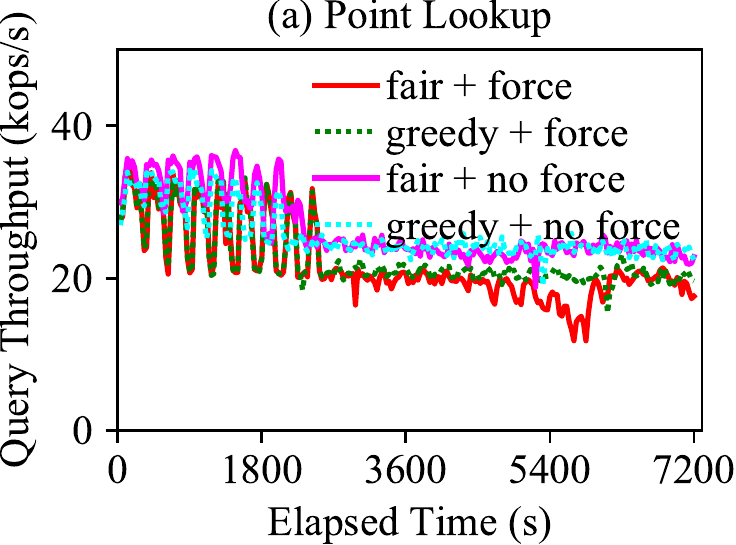}
		\end{minipage}
		\hfil
		\begin{minipage}[t]{.32\textwidth}
			\centering
			\includegraphics[width=\scale\linewidth]{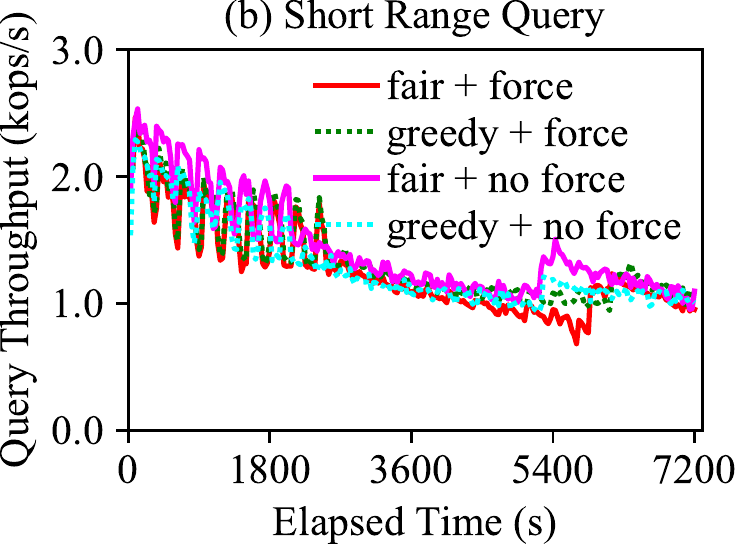}
		\end{minipage}
		\hfil
		\begin{minipage}[t]{.32\textwidth}
			\centering
			\includegraphics[width=\scale\linewidth]{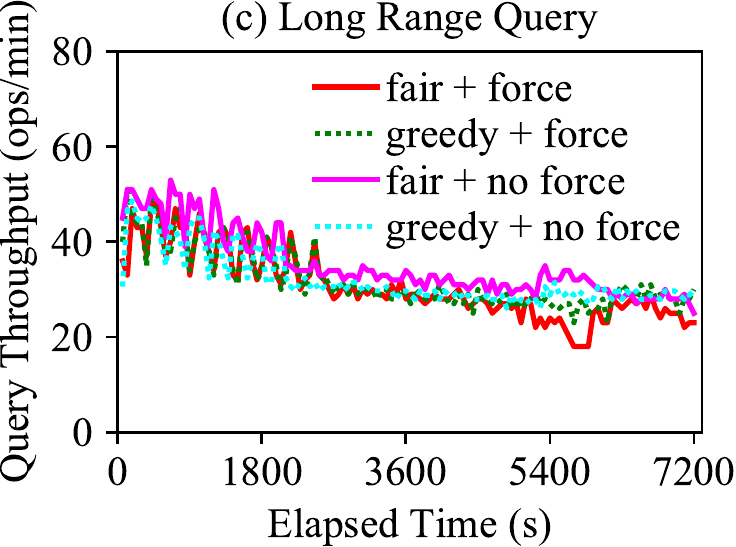}
		\end{minipage}
		\hfil
		\caption{Instantaneous Query Throughput of Leveling Merge Policy}
		\label{fig:expr-query-level-extend}
	\end{figure*}

	\begin{figure*}
	\centering
	\begin{minipage}[t]{.32\textwidth}
		\centering
		\includegraphics[width=\scale\linewidth]{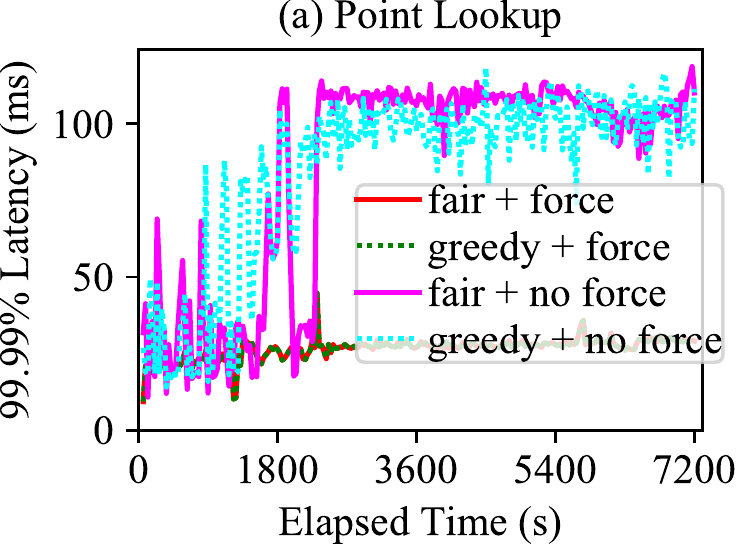}
	\end{minipage}
	\hfil
	\begin{minipage}[t]{.32\textwidth}
		\centering
		\includegraphics[width=\scale\linewidth]{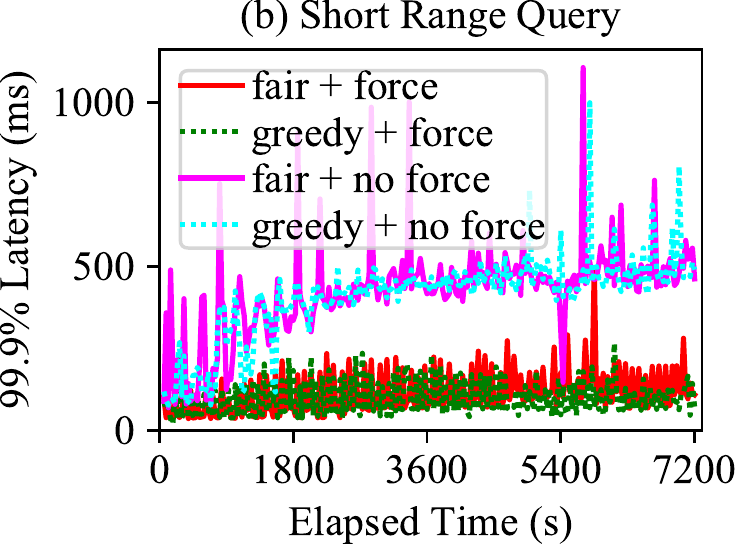}
	\end{minipage}
	\hfil
	\begin{minipage}[t]{.32\textwidth}
		\centering
		\includegraphics[width=\scale\linewidth]{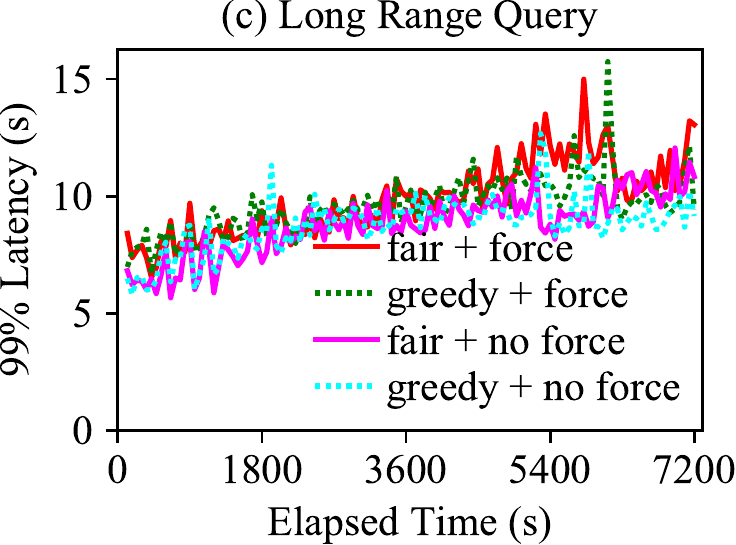}
	\end{minipage}
	\hfil
	\caption{Percentile Query Latencies of Leveling Merge Policy}
	\label{fig:expr-query-level-latency}
\end{figure*}

}

\textbf{Benefits of Processing Writes As Quickly As Possible.}
We further evaluated the benefit of processing writes as quickly as possible.
We used the leveling merge policy with a bursty data arrival process
that alternates between a normal arrival rate of 2000 records/s for 25 minutes and a high arrival rate of 8000 records/s for 5 minutes.
We evaluated two variations of the greedy scheduler.
The first variation processes writes as quickly as possible (denoted as ``No Limit''), as we did before.
The second variation enforces a maximum in-memory write rate of 4000 records/s (denoted as ``Limit'') to avoid write stalls.

The instantaneous write throughput and the percentile write latencies of the two variations are
shown in Figures \ref{fig:expr-level-burst}a and \ref{fig:expr-level-burst}b respectively.
As \reffigure{fig:expr-level-burst}a shows, delaying writes avoids write stalls and the resulting write throughput is more stable over time.
However, this causes larger write latencies (\reffigure{fig:expr-level-burst}b) since delayed writes must be queued.
In contrary, writing as quickly as possible causes occasional write stalls but still minimizes overall write latencies.
This confirms our previous analysis that processing writes as quickly as possible minimizes write latencies.

\textbf{Impact on Query Performance.}
Finally, since the point of having data is to query it,
we evaluated the impact of the fair and greedy schedulers on concurrent query performance.
We evaluated three types of queries, namely point lookups, short scans, and long scans.
A point lookup accesses 1 record given a primary key.
A short scan query accesses 100 records and a long scan query accesses 1 million records.
In each experiment, we executed one type of query concurrently with {concurrent updates with constant arrival rates as before}.
To maximize query performance while ensuring that LSM flush and merge operations receive enough {I/O bandwidth},
we used 8 query threads for point lookups and short scans and used 4 query threads for long scans.
\longonly{We also evaluated the impact of forcing SSD writes regularly on query performance.
For this purpose, we included the variations of the fair and greedy schedulers that only force SSD writes when a merge completes.}

\shortonly{
The instantaneous query throughput under tiering and leveling is depicted in Figures \ref{fig:expr-query-tier} and \ref{fig:expr-query-level} respectively.
{Due to space limitations, we omit the average query latency, which can be computed by dividing the number of query threads by the query throughput. As the results show, leveling has similar point lookup throughput to tiering because Bloom filters are able to filter out most unnecessary I/Os,
	but it has much better range query throughput than tiering.}
The greedy scheduler always improves query performance by minimizing the number of components.
Among the three types of queries, point lookups and short scans benefit more from the greedy scheduler
since these two types of queries are more sensitive to the number of disk components.
In contrast, long scans incur most of their I/O cost at the largest level.
Moreover, tiering benefits more from the greedy scheduler than leveling
because tiering has more disk components.
Note that with leveling, there is a drop in query throughput under the fair scheduler at around 5400s,
even though there is little difference in the number of disk components between the fair and greedy schedulers.
This drop is caused by write stalls during that period, as was seen in the instantaneous write throughput of \reffigure{fig:expr-level-open}a.
After the LSM-tree recovers from write stalls, it attempts to write as much data as possible to catch up,
which negative impacts query performance.

{In~\cite{lsm-stability-extend}, we additionally evaluated the impact of forcing SSD writes regularly on query performance.
Even though this optimization has some small negative impact on query throughput, it significantly reduces the percentile latencies of small queries, e.g., point lookups and small scans, by 5x to 10x because the large forces at the end of merges are instead broken down into many smaller ones.}
}

\longonly{
The instantaneous query throughput under the tiering and leveling merge policies is depicted in \reffigure{fig:expr-query-tier-extend}
and \reffigure{fig:expr-query-level-extend} respectively.
The corresponding percentile query latencies are shown in \reffigure{fig:expr-query-level-latency} and \reffigure{fig:expr-query-tier-latency} respectively.
For point lookups and short scans, the query throughput is averaged over 30-second windows.
For long scans, the query throughput is averaged over 1-minute windows.
As the results show, leveling has similar point lookup throughput to tiering because Bloom filters can filter out most unnecessary I/Os,
	but it has much better range query throughput than tiering.
Moreover, the greedy scheduler always improves query performance by minimizing the number of components.
Among the three types of queries, point lookups and short scans benefit more from the greedy scheduler
since these two types of queries are more sensitive to the number of disk components.
In contrast, long scans incur most of their I/O cost at the largest level.
Moreover, the tiering merge policy benefits more from the greedy scheduler than the leveling merge policy
since the performance difference between the greedy and fair schedulers is larger under the tiering merge policy.
This is because the tiering merge policy has more disk components than the leveling merge policy.
Note that under the leveling merge policy, there is a drop in query throughput under the fair scheduler at around 5400s,
even though there is little difference in the number of disk components between the fair and greedy scheduler.
This drop is caused by write stalls during that period, as indicated by the instantaneous write throughput of \reffigure{fig:expr-level-open}.
After the LSM-tree recovers from write stalls, it attempts to write as much data as possible in order to catch up,
which results in a lower query throughput.

Forcing SSD writes regularly has some slight negative impact on query throughput, but it significantly reduces the percentile query latencies.
The reason is that disk components must be forced to disk in the end of merges. Without forcing SSD writes regularly, the large disk forces will significantly impact the query latencies.
}

\longonly{
\subsection{Tiering in Practice}
Existing LSM-based systems, such as BigTable~\cite{bigtable} and HBase~\cite{hbase},
use a slight variation of the tiering merge policy discussed in the literature.
This variation, often referred as the \emph{size-tiered} merge policy, does not organize components into levels explicitly
but simply schedules merges based on the sizes of disk components.
This policy has three important parameters, namely the size ratio $T$, the minimum number of components to merge $min$,
and the maximum number of components to merge $max$.
It merges a sequence of components, whose length is at least $min$, when the total size of the sequence's
the younger components is $T$ times larger than that of the oldest component in the sequence.
It also seeks to merge as many components as possible at once until $max$ is reached.
Concurrent merges can also be performed.
For example, in HBase~\cite{hbase}, each execution of the size-tiered merge policy
will always examine the longest prefix of the component sequence in which no component is being merged.

An example of the size-tiered merge policy is shown in \reffigure{fig:size-tiered},
where each disk component is labeled with its size.
Let the size ratio be $1.2$ and the minimum and maximum number of components per merge be 2 and 4 respectively.
Suppose initially that no component is being merged.
The first of execution the size-tiered merge policy starts from the oldest component, labeled 100GB.
However, no merge is scheduled since this component is too large.
It then examines the next component, labeled 10GB, and schedules a merge operation for the 4 components labeled from 10GB to 5GB.
The next execution of the size-tiered merge policy starts from the component labeled 1GB, 
and it schedules a merge for the 3 components labeled from 128MB to 64MB.

\begin{figure}
	\centering
	\includegraphics[width=\linewidth]{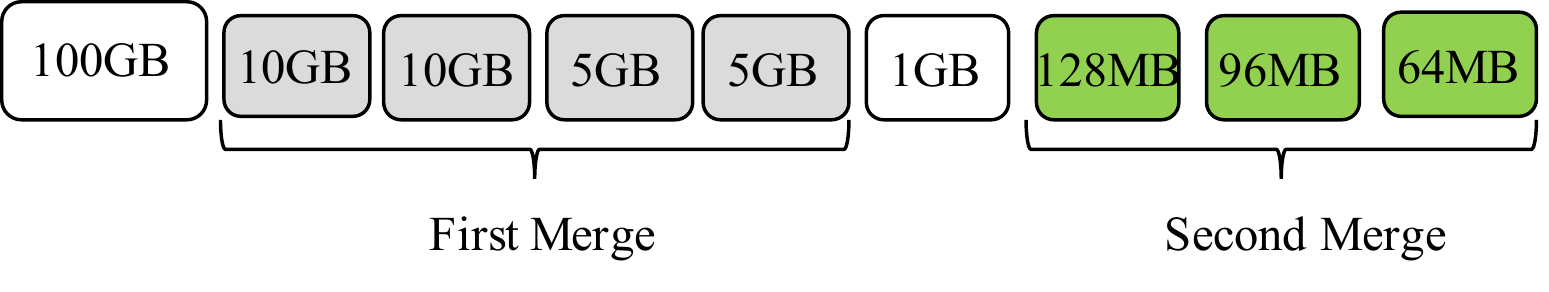}
	\caption{Example of Size-Tiered Merge Policy}
	\label{fig:size-tiered}
\end{figure}

To evaluate the write stalls of the size-tiered merge policy, we repeated the experiments using our two-phase approach.
In our evaluation, the size ratio was set at 1.2, which is the default value in HBase~\cite{hbase},
and the minimum and maximum mergeable components were set at 2 and 10 respectively.
The maximum tolerated disk components parameter was set at 50.

During the testing phase, the maximum write throughput measured by using the fair scheduler was 17,008 records/s.
Then during the running phase,
we used a constant data arrival process based on 95\% of this maximum throughput to evaluate write stalls.
The instantaneous write throughput of the LSM-tree and the number of disk components over time
are shown in Figures \ref{fig:expr-size-tiered}a and \ref{fig:expr-size-tiered}b respectively.
As one can see, write stalls have occurred under the fair scheduler.
Moreover, even though the greedy scheduler avoids write stalls, its number of disk components keeps increasing over time.
This result indicates that the maximum write throughput measured during the testing phase is not sustainable.

\begin{figure}
	\begin{subfigure}[t]{.235\textwidth}
		\centering
		\includegraphics[width=\linewidth]{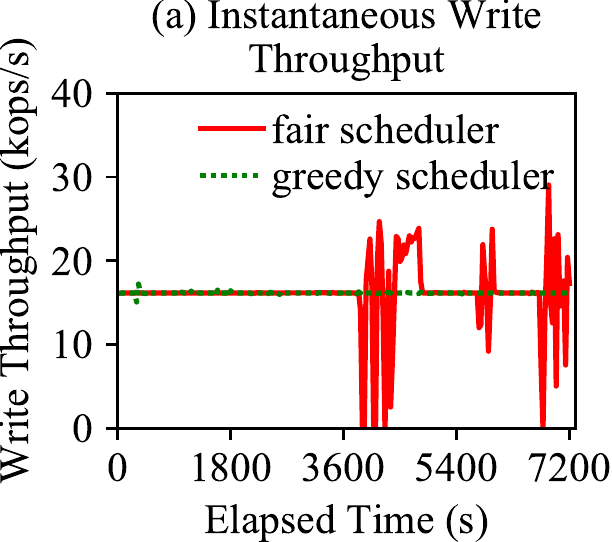}
	\end{subfigure}
	\hfil
	\begin{subfigure}[t]{.235\textwidth}
		\centering
		\includegraphics[width=\linewidth]{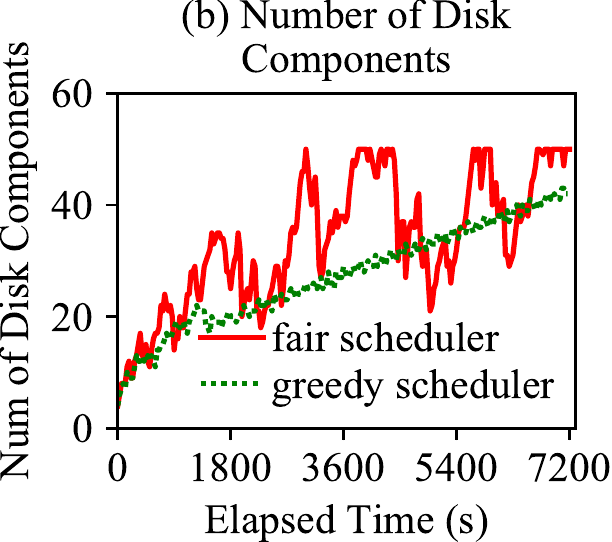}
	\end{subfigure}
	\caption{Running Phase of Size-Tiered Merge Policy (95\% Load)}
	\label{fig:expr-size-tiered}
\end{figure}

This problem is caused by the non-determinism of the size-tiered merge policy since it tries to merge as many disk components as possible.
This behavior impacts the maximum write throughput of the LSM-tree.
During the testing phase, when writes are often blocked because of too many disk components,
this merge policy tends to merge more disk components at once, which then leads to a higher write throughput.
However, during the running phase, when writes arrive steadily, this merge policy tends to schedule smaller merges as flushed components accumulate.
For example, during the running phase of this experiment, 55 long merges that involved 10 components were scheduled,
but only 24 long merges were scheduled under the fair scheduler during the running phase.
Even worse, 99.76\% of the scheduled merges under the greedy scheduler involved no more than 4 components
since large merges were starved.

To address problem and to minimize write stalls, the arrival rate must be reduced.
However, finding the maximum ``stall free'' arrival rate is non-trivial due to the non-determinism of the size-tiered merge policy.
Instead, we propose a simple and conservative solution to avoid write stalls.
During the testing phase, we propose to measure the lower bound write throughput by always merging the minimum number of disk components.
This write throughput will serve as a baseline of the arrival rate.
During runtime, the size-tiered merge policy can merge more disk components to
dynamically increase its write throughput to minimize stalls.

We repeated the previously experiments based on this solution.
During the testing phase, the merge policy always merged 2 disk components,
which resulted in a lower maximum write throughput of 8,863 records/s.
We then repeated the running phase based on this throughput.
Figures \ref{fig:expr-size-tiered-fixed}a and \ref{fig:expr-size-tiered-fixed}b
show the instantaneous write throughput and the number of disk components over time respectively during the running phase.
In this case, both schedulers exhibit no write stalls and the number of disk components is more stable over time.
Moreover, the greedy scheduler still slightly reduces the number of disk components.


\begin{figure}
	\begin{subfigure}[t]{.235\textwidth}
		\centering
		\includegraphics[width=\linewidth]{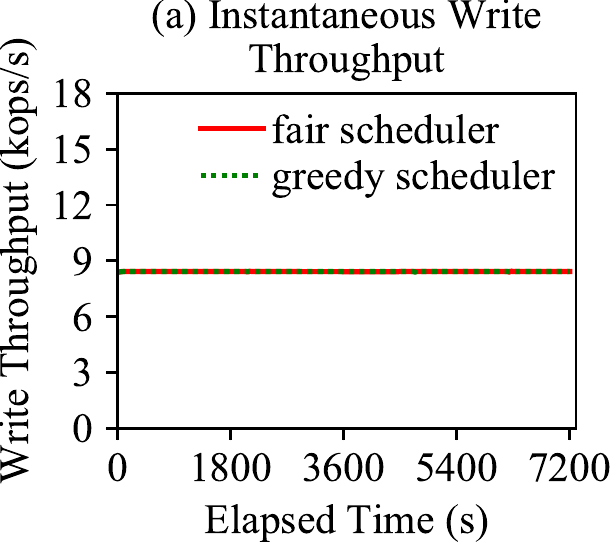}
	\end{subfigure}
	\hfil
	\begin{subfigure}[t]{.235\textwidth}
		\centering
		\includegraphics[width=\linewidth]{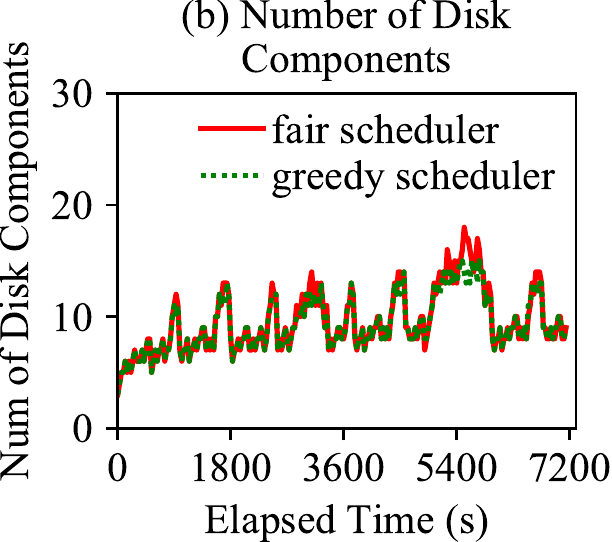}
	\end{subfigure}
	\caption{Running Phase of Size-Tiered Merge Policy with the Proposed Solution}
	\label{fig:expr-size-tiered-fixed}
\end{figure}
}

\section{Partitioned Merges}
\label{sec:partitioning}
We now examine the write stall behavior of partitioned LSM-trees using our two-phase approach.
In a partitioned LSM-tree, a large disk component is range-partitioned into multiple small {files}
and each merge operation only processes a small number of {files} with overlapping ranges.
Since merges always happen immediately once a level is full,  a single-threaded scheduler could be sufficient to minimize write stalls.
In the reminder of this section, we will evaluate LevelDB's single-threaded scheduler.

\subsection{LevelDB's Merge Scheduler}
LevelDB's merge scheduler is single-threaded.
It computes a score for each level and selects the level with the largest score to merge.
Specifically, the score for {Level} 0 is computed as the total number of flushed components divided by the minimum number of flushed components to merge.
For a partitioned level (1 and above), its score is defined as the total size of all {files} at this level divided by the configured maximum size.
A merge operation is scheduled if the largest score is at least 1, which means that the selected level is full.
If a partitioned level is chosen to merge, LevelDB selects the next {file} to merge in a round-robin way.

LevelDB only restricts the number of flushed components at {Level} 0.
By default, the minimum number of flushed components to merge is 4.
The processing of writes will be slowed down or stopped of the number of flushed component reaches 8 and 12 respectively.
Since we have already shown in \refsection{sec:full-merge-interact-writes} that processing writes as quickly as possible reduces write latencies, we will only use the stop threshold (12) in our evaluation.

\textbf{Experimental Evaluation.}
We have implemented LevelDB's partitioned leveling merge policy and its merge scheduler inside AsterixDB for evaluation.
Similar to LevelDB, the minimum number of flushed components to merge was set at 4 and the stop threshold was set at 12 components.
{Unless otherwise noted, the maximum size of each file was set at 64MB.}
The memory component size was set at 128MB and the base size of {Level} 1 was set at 1280MB.
The size ratio was set at 10.
For the experimental dataset with 100 million records, this results in a 4-level LSM-tree where the largest level is nearly full.
To minimize write stalls caused by flushes, we used two memory components and a separate flush thread.
{We further evaluated the impact of two widely used merge selection strategies on write stalls.}
The round-robin strategy chooses the next {file} to merge in a round-robin way.
The choose-best strategy~\cite{partial-merge2017} chooses the {file} with the fewest overlapping files at the next level.

We used our two-phase approach to evaluate this partitioned LSM-tree design.
The instantaneous write throughput during the testing phase is shown in \reffigure{fig:expr-partition}a,
where the write throughput of both strategies decreases over time due to more frequent stalls.
{Moreover, under the uniform update workload,
the alternative selection strategies have little impact on the overall write throughput, as reported in~\cite{lsm-model2016}.}
During the testing phase, we used a constant arrival process to evaluate write stalls.
The instantaneous write throughput of both strategies is shown in \reffigure{fig:expr-partition}b.
As the result shows, in both cases write stalls start to occur after time 6000s.
This suggests that the measured write throughput during the testing phase is not sustainable.

\begin{figure}
	\begin{subfigure}[t]{.235\textwidth}
		\centering
		\includegraphics[width=\linewidth]{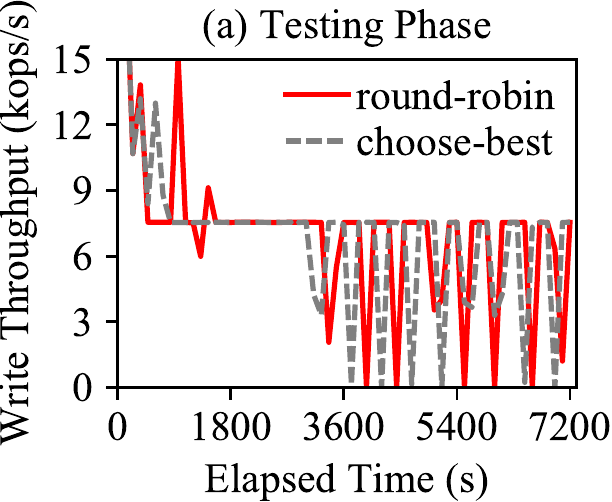}
	\end{subfigure}
	\hfil
	\begin{subfigure}[t]{.235\textwidth}
		\centering
		\includegraphics[width=\linewidth]{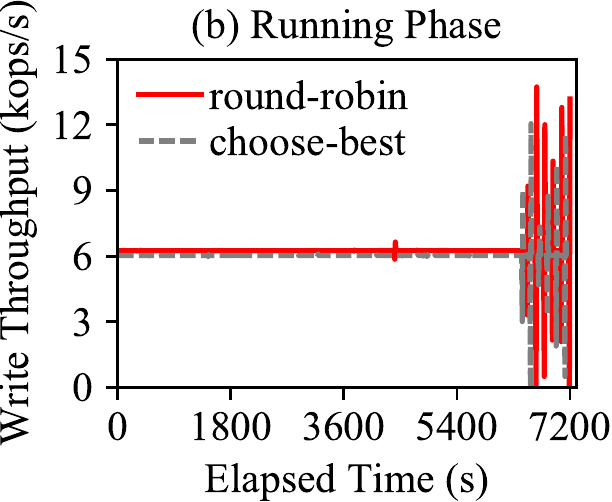}
	\end{subfigure}
	\vspace{-0.2in}
	\caption{Instantaneous Write Throughput under Two-Phase Evaluation of Partitioned LSM-tree}
	\label{fig:expr-partition}
\end{figure}

\subsection{Measuring Sustainable Write Throughput}
One problem with LevelDB's score-based merge scheduler is that it merges as many components at {Level} 0 as possible at once.
To see this, suppose that the minimum number of mergeable components at {Level} $L_0$ is $T_0$ 
and that the maximum number of components at {Level} $0$ is $T'_0$.
During the testing phase, where writes pile up as quickly as possible,
the merge scheduler tends to merge the maximum possible number of components $T'_0$ instead of just $T_0$ at once.
Because of this, the LSM-tree will eventually transit from the expected shape (\reffigure{fig:partitioning}a)
to the actual shape (\reffigure{fig:partitioning}b), where $T$ is the size ratio of the partitioned levels.
Note that the largest level is not affected since its size is determined by the number of unique entries, which is relatively stable.
Even though this elastic design dynamically increases the processing rate as needed,
it has the following problems.

{
\textbf{Unsustainable Write Throughput.}
The measured maximum write throughput is based on merging $T'_0$ flushed components at {Level} $0$
at once. However, this is likely to cause write stalls during the running phase since flushes cannot further proceed.

\textbf{Suboptimal Trade-Offs.}
The LSM-tree structure in \reffigure{fig:partitioning}b is no longer making optimal performance trade-offs since
the size ratios between its adjacent levels are not the same anymore~\cite{lsm1996}.
By adjusting the sizes of intermediate levels so that adjacent levels have the same size ratio,
one can improve both write throughput and space utilization without affecting query performance.

\textbf{Low Space Utilization.}
One motivation for industrial systems to adopt partitioned LSM-trees is their higher space utilization~\cite{rocksdb-space2017}.
However, the LSM-tree depicted in \reffigure{fig:partitioning}b violates this performance guarantee because the
ratio of wasted space increases from $1/T$ to $T'_0/T_0 \cdot 1/T$.
}

\begin{figure}
		\centering
		\includegraphics[width=\linewidth]{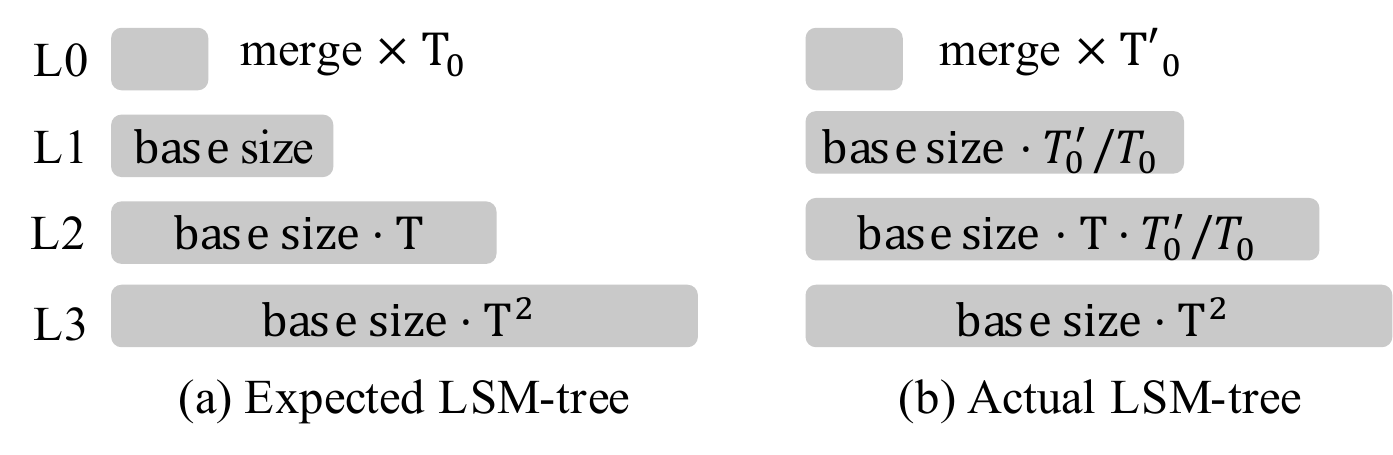}
		\vspace{-0.25in}
		\caption{Problem of Score-Based Merge Scheduler}
		\label{fig:partitioning}
\end{figure}

{Because of these problems, the measured maximum write throughput cannot be used in the long-term.}
We propose a simple solution to address these problems.
During the testing phase, we always merge exactly $T_0$ components at {Level} 0.
This ensures that merge preferences will be given equally to all levels so that
the LSM-tree will stay in the expected shape (\reffigure{fig:partitioning}a).
Then, during the running phase, the LSM-tree can elastically merge more components at {Level} 0 as needed to absorb write bursts.

To verify the effectiveness of the proposed solution, we repeated the previous experiments on the partitioned LSM-tree.
During the testing phase, the LSM-tree always merged 4 components at {Level} 0 at once.
The measured instantaneous write throughput is shown in \reffigure{fig:expr-partition-fixed}a,
which is {30\%} lower than that of the previous experiment.
During the running phase, we used a constant arrival process based on this lower write throughput.
The resulting instantaneous write throughput is shown in \reffigure{fig:expr-partition-fixed}b,
where the LSM-tree successfully maintains a sustainable write throughput without any write stalls,
which in turn results in low write latencies (not shown in the figure).
This confirms that LevelDB's single-threaded scheduler is sufficient to minimize write stalls,
given that a single merge thread can fully utilize the {I/O bandwidth} budget.

\begin{figure}
	\centering
	\begin{subfigure}[t]{.235\textwidth}
		\centering
		\includegraphics[width=\textwidth]{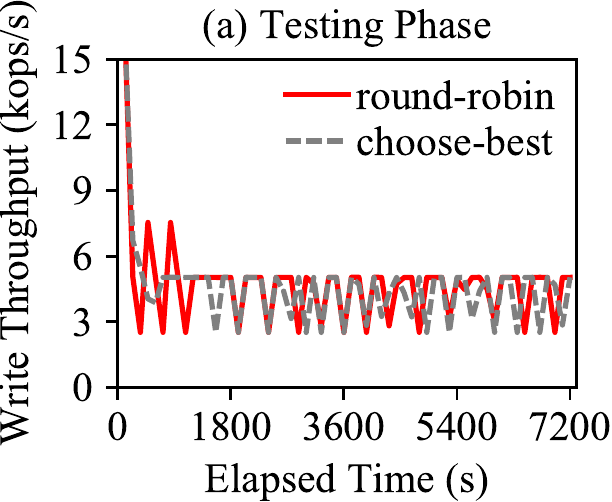}
	\end{subfigure}
	\hfil
	\begin{subfigure}[t]{.235\textwidth}
		\centering
		\includegraphics[width=\linewidth]{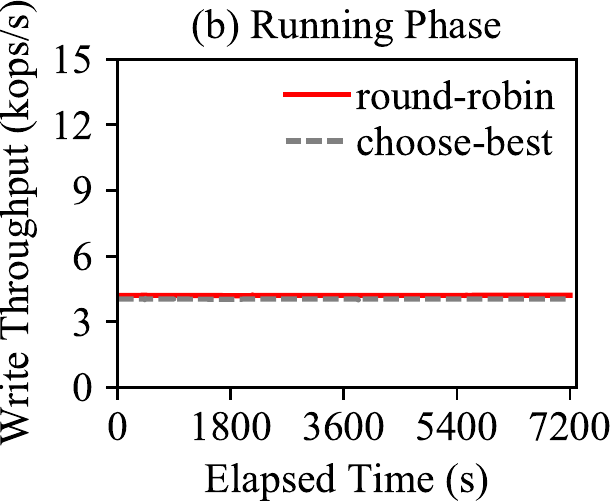}
	\end{subfigure}
	\vspace{-0.2in}
	\caption{Instantaneous Write Throughput under Two-Phase Evaluation of Partitioned LSM-tree with the Proposed Solution}
	\label{fig:expr-partition-fixed}
\end{figure}

{
After fixing the unsustainable write throughput problem of LevelDB, 
we further evaluated the impact of partition size on the write stalls of partitioned LSM-trees.
In this experiment, we varied the size of each partitioned file from 8MB to 32GB so that partitioned merges effectively transit into full merges.
The maximum write throughput during the running phase and the 99th percentile write latencies during the testing phase
are shown in Figures \ref{fig:expr-partition-size}a and \ref{fig:expr-partition-size}b respectively.
Even though the partition size has little impact on the overall write throughput, a large partition size can cause large write latencies
since we have shown in \refsection{sec:full-merge} that a single-threaded scheduler is insufficient to minimize write stalls for full merges.
Most implementations of partitioned LSM-trees today already choose a small partition size to bound the temporary space occupied by merges.
We see here that one more reason to do so is to minimize write stalls under a single-threaded scheduler.
}

\begin{figure}
	\centering
	\begin{subfigure}[t]{.227\textwidth}
		\centering
		\includegraphics[width=\linewidth]{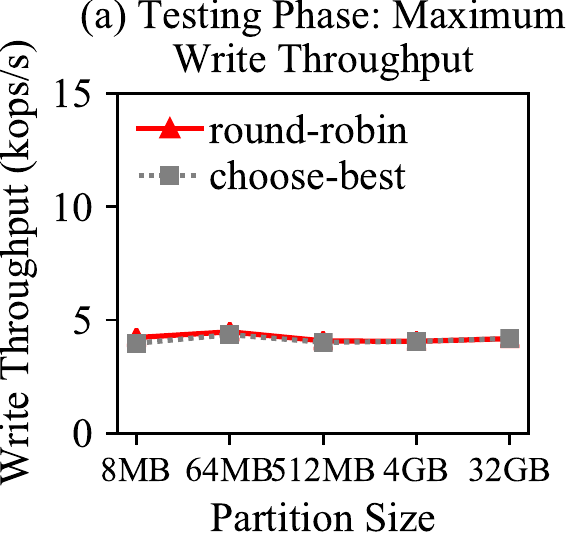}
	\end{subfigure}
	\hfil
	\begin{subfigure}[t]{.243\textwidth}
		\centering
		\includegraphics[width=\linewidth]{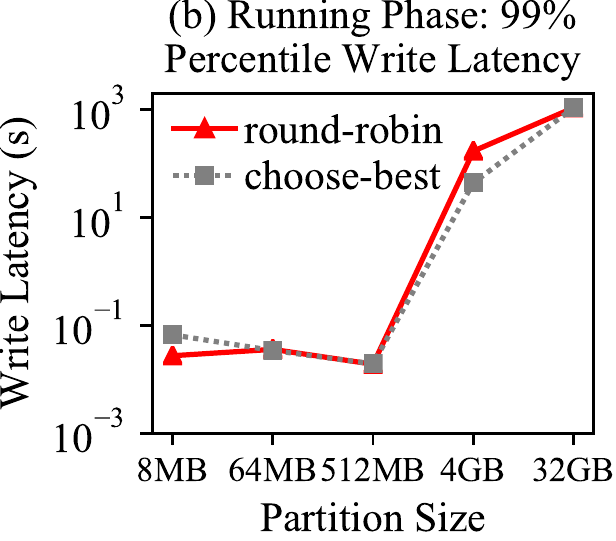}
	\end{subfigure}
	\vspace{-0.2in}
	\caption{{Impact of Partition Size on Write Stalls}}
	\label{fig:expr-partition-size}
\end{figure}

\longonly{
\section{Extension: Secondary Indexes}
\label{sec:secondary-indexes}
We now extend our two-phase approach to evaluate LSM-based datasets in the presence of secondary indexes.
We first discuss two secondary index maintenance strategies used in practical systems, followed by the experimental evaluation and analysis.

\begin{figure}
	\centering
	\begin{subfigure}[t]{.235\textwidth}
		\centering
		\includegraphics[width=\linewidth]{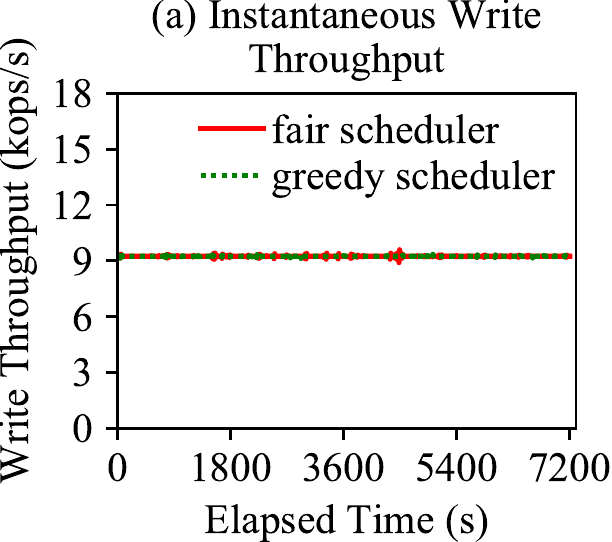}
	\end{subfigure}
	\hfil
	\begin{subfigure}[t]{.235\textwidth}
		\centering
		\includegraphics[width=\linewidth]{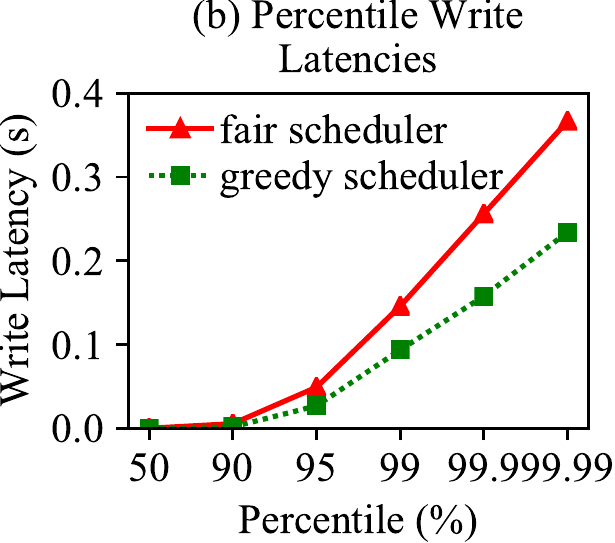}
	\end{subfigure}
	\caption{Running Phase of Lazy Strategy}
	\label{fig:expr-lazy-running}
\end{figure}

\begin{figure}
	\centering
	\begin{subfigure}[t]{.235\textwidth}
		\centering
		\includegraphics[width=\linewidth]{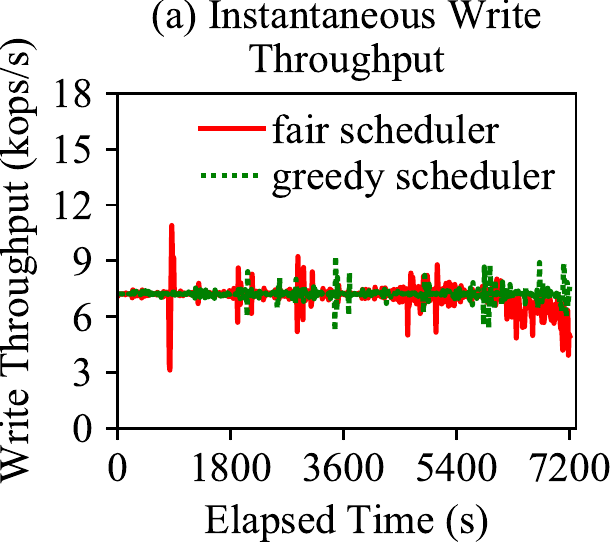}
	\end{subfigure}
	\hfil
	\begin{subfigure}[t]{.235\textwidth}
		\centering
		\includegraphics[width=\linewidth]{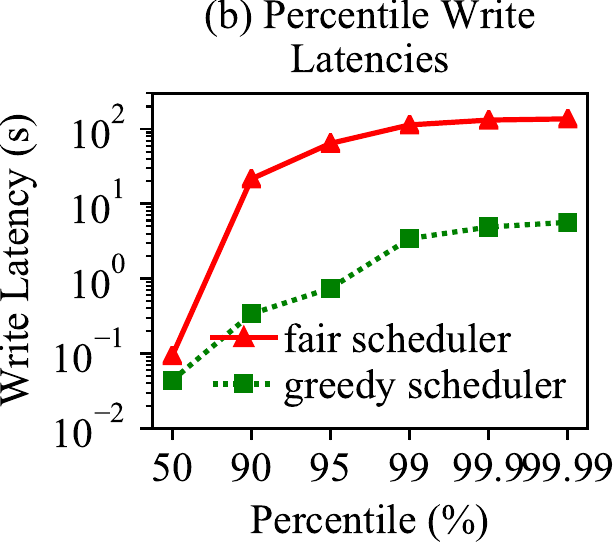}
	\end{subfigure}
	\caption{Running Phase of Eager Strategy}
	\label{fig:expr-eager-running}
\end{figure}

\subsection{Secondary Index Maintenance}
An LSM-based storage system often contains a primary index plus multiple secondary indexes for a given dataset~\cite{lsm-storage2019}.
The primary index stores the records indexed by their keys,
while each secondary index stores the mapping from secondary keys to primary keys.
During data ingestion, secondary indexes must be properly maintained to ensure correctness.
In the primary LSM-tree, writes (inserts, deletes, and updates) can be added blindly to memory
since entries with identical keys will be reconciled by queries automatically.
However, this mechanism does not work for secondary indexes since the value of a secondary index key might change.
Thus, in addition to adding the new entry to the secondary index, the old entry (if any) must be cleaned as well.
We now discuss two secondary index maintenance strategies used in practice~\cite{lsm-storage2019}.

The \emph{eager} index maintenance strategy performs a point lookup to fetch the old record during the ingestion time.
If the old record exists, anti-matter entries are produced to cleanup its secondary indexes.
The new record is then added to the primary index and all secondary indexes.
In an update-heavy workload, these point lookups can become the ingestion bottleneck instead of the LSM-tree write operations.

The \emph{lazy} index maintenance strategy does not cleanup secondary indexes during the ingestion time.
Instead, it only adds the new entry into secondary indexes without any point lookups.
Secondary indexes are then cleaned up in the background either when merging the primary index components~\cite{deli2015} or when merging the secondary index components~\cite{lsm-storage2019}.
Evaluating different secondary index cleanup methods is beyond the scope of this work.
Instead, we choose to evaluate the lazy strategy without cleaning up secondary indexes.

\begin{figure}
	\centering
	\includegraphics[width=0.65\linewidth]{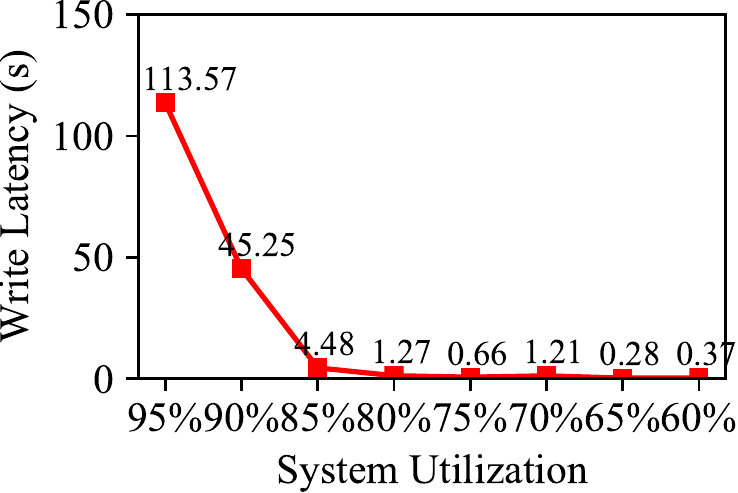}
	\caption{99\% Percentile Write Latencies under Eager Strategy with Varying Utilization}
	\label{fig:expr-eager-util-write-latency}
\end{figure}

\begin{figure*}
	\centering
	\begin{minipage}[t]{.24\textwidth}
		\centering
		\includegraphics[width=\linewidth]{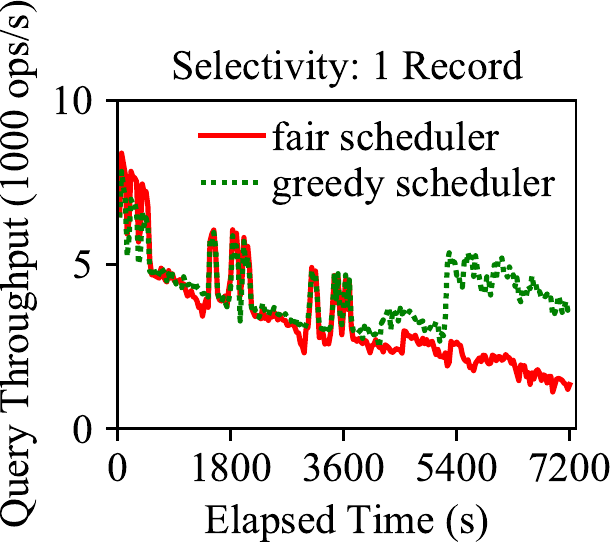}
	\end{minipage}
	\hfil
	\begin{minipage}[t]{.23\textwidth}
		\centering
		\includegraphics[width=\linewidth]{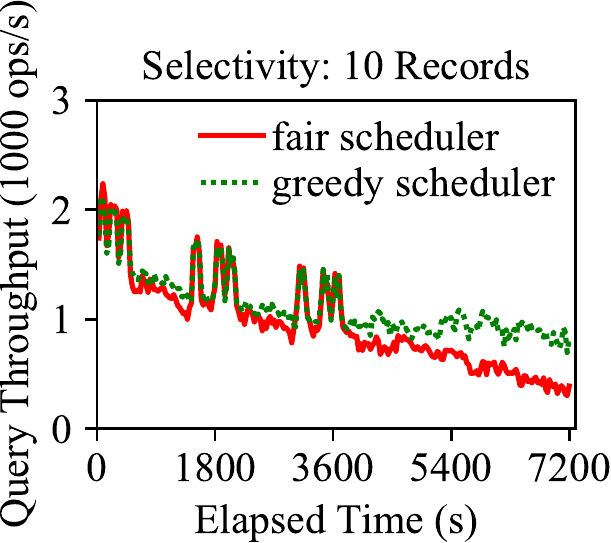}
	\end{minipage}
	\hfil
	\begin{minipage}[t]{.24\textwidth}
		\centering
		\includegraphics[width=\linewidth]{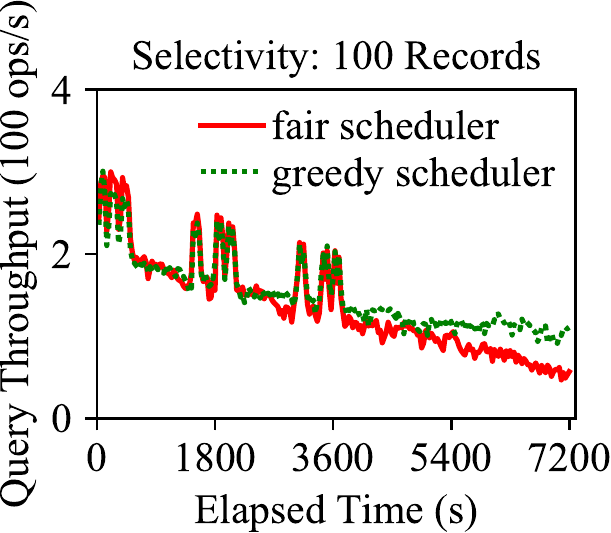}
	\end{minipage}
	\hfil
	\begin{minipage}[t]{.24\textwidth}
		\centering
		\includegraphics[width=\linewidth]{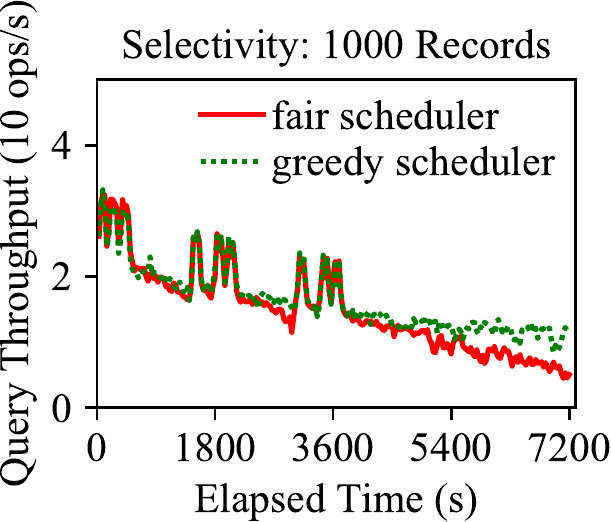}
	\end{minipage}
	\caption{Instantaneous Query Throughput of Lazy Strategy}
	\label{fig:expr-lazy-query}
\end{figure*}

\begin{figure*}
	\centering
	\begin{minipage}[t]{.24\textwidth}
		\centering
		\includegraphics[width=\linewidth]{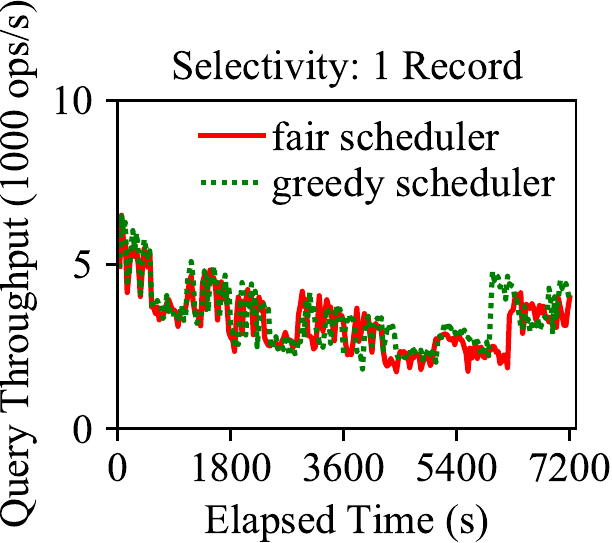}
	\end{minipage}
	\hfil
	\begin{minipage}[t]{.23\textwidth}
		\centering
		\includegraphics[width=\linewidth]{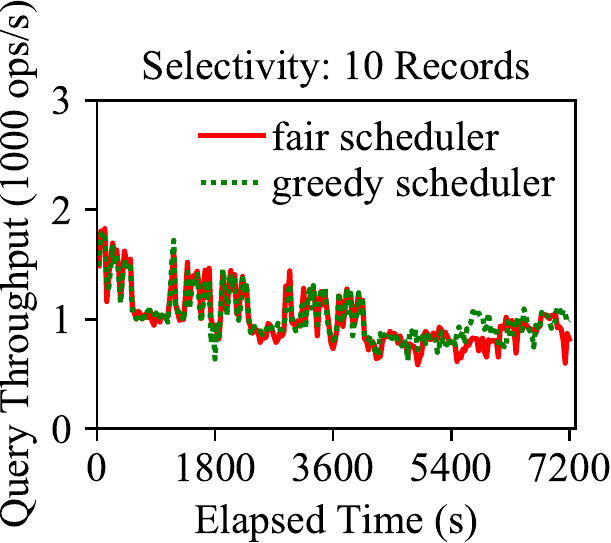}
	\end{minipage}
	\hfil
	\begin{minipage}[t]{.24\textwidth}
		\centering
		\includegraphics[width=\linewidth]{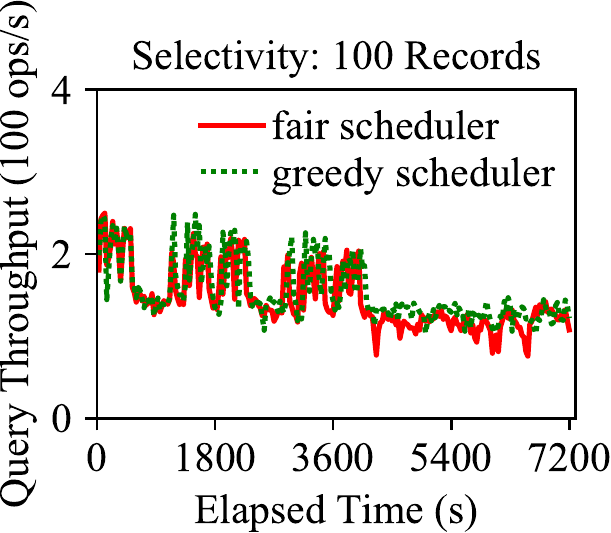}
	\end{minipage}
	\hfil
	\begin{minipage}[t]{.24\textwidth}
		\centering
		\includegraphics[width=\linewidth]{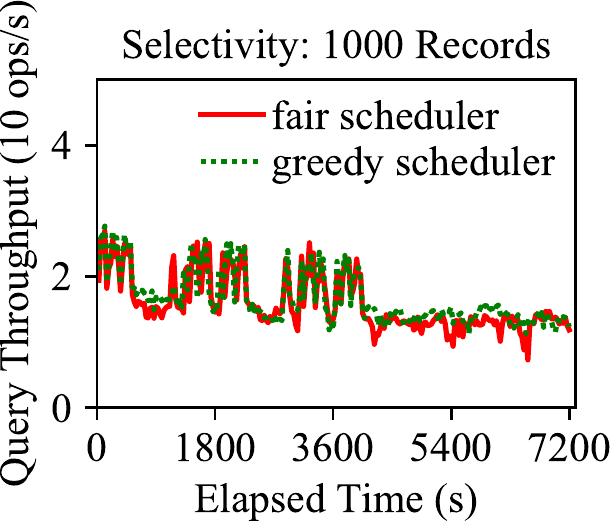}
	\end{minipage}
	\caption{Instantaneous Query Throughput with Eager Strategy}
	\label{fig:expr-eager-query}
\end{figure*}

\subsection{Experimental Evaluation}
\textbf{Experiment Setup.}
In this set of experiments, we modified the YCSB benchmark to allow us incorporate secondary indexes and formulate secondary index queries.
Specifically, we generated records with multiple fields with
each secondary field value randomly following a uniform distribution based on the total number of base records.
We built two secondary indexes in our experiment.
The primary index and the two secondary indexes all used the tiering merge policy with size ratio 3.

In this set of experiments, we evaluated two merge schedulers, namely fair and greedy.
Each LSM-tree is merged independently with a separate merge scheduler instance.
However, these LSM-trees shared the same memory budget 128MB for each memory component
and the {I/O} bandwidth budget of 100MB/s.
We also evaluated two index maintenance strategies, namely eager and lazy.
For the eager strategy, we used 8 writer threads to maximize the point lookup throughput.
For the lazy strategy, 1 writer thread was sufficient to reach the maximum write throughput since there were no point lookups during data ingestion.

\textbf{Testing Phase.}
We first measured the maximum write throughput of the lazy and eager strategies using the fair scheduler during the testing phase.
The maximum write throughput was 9,731 records/s for the lazy strategy and 7,601 records/s for the eager strategy.
(The eager strategy results in a slightly lower write throughput because
it has to cleanup secondary indexes using point lookups.)

\textbf{Running Phase.}
During the running phase, we used constant data arrivals to evaluate write stalls.
The instantaneous write throughput and percentile write latencies for the lazy and eager strategies are shown in Figures \ref{fig:expr-lazy-running}
and \ref{fig:expr-eager-running} respectively.
The lazy strategy exhibits a relatively stable write throughput (\reffigure{fig:expr-lazy-running}a) and lower write latencies (\reffigure{fig:expr-lazy-running}b), which is similar to the single LSM-tree case.
However, under the eager strategy, there are regular fluctuations in the write throughput (\reffigure{fig:expr-eager-running}a),
results in larger write latencies (\reffigure{fig:expr-eager-running}b).
This is because the write throughput of the eager strategy is bounded by point lookups in this experiment,
and the point lookup throughput inherently varies due to ongoing disk activities and the number of disk components.
Based on queuing theory~\cite{queuing2013}, the system utilization must be reduced to minimize the write latency.
Moreover, the greedy scheduler still has lower write latencies due to its
minimizing the number of disk components to improve point lookup performance.

Since the eager strategy results in large percentile write latencies under a high data arrival rate,
we further carried out another experiment to evaluate the percentile write latencies under different system utilizations, that is, different data arrival rates.
The resulting 99\% percentile write latencies under various utilizations are shown in \reffigure{fig:expr-eager-util-write-latency}.
As the result shows, the write latency becomes relatively small once the utilization is below 80\%.
This is much smaller than the utilization used in our previous experiments, which was 95\%.
This result also confirms that because of the inherent variance of the point lookup throughput, one must reduce the data arrival rate, that is, the system utilization, to achieve smaller write latencies.

\textbf{Secondary Index Queries.}
Finally, we evaluated the impact of different merge schedulers and maintenance strategies on the performance of secondary index queries.
We used 8 query threads to maximize query throughput.
Each secondary index query first scans the secondary index to fetch primary keys,
which are then sorted and used to fetch records from the primary index.
We varied the query selectivity from 1 record to 1000 records so that
the performance bottleneck eventually shifts from secondary index scans to primary index lookups.

The instantaneous query throughput for various query selectivities under the lazy and eager strategy is shown in Figures \ref{fig:expr-lazy-query} and \ref{fig:expr-eager-query} respectively.
The query throughput is averaged over each 30-second windows.
In general, the greedy scheduler improves secondary index query performance under all query selectivities
since it reduces the number of disk components for both the primary index and secondary indexes.
The improvement is less significant under the eager strategy since the arrival rate is lower.

To summarize, under the lazy strategy, an LSM-based dataset with multiple secondary indexes has similar performance characteristics
to the single LSM-tree case, because this can be viewed as a simple extension to multiple LSM-trees.
The greedy scheduler also improves query performance by minimizing the number of disk components as before.
However, under the eager strategy, the point lookups actually become the ingestion bottleneck instead of LSM-tree write operations.
This not only reduces the overall write throughput, but further causes larger write latencies due to the inherent variance of the point lookup throughput.
}

\section{Lessons and Insights}
\label{sec:summary}
Having studied and evaluated the write stall problem for various LSM-tree designs,
here we summarize the lessons and insights observed from our evaluation.

\textbf{The LSM-tree's write latency must be measured properly.}
The out-of-place update nature of LSM-trees has introduced the write stall problem.
Throughout our evaluation, we have seen cases where one can obtain a higher but unsustainable write throughput.
For example, the greedy scheduler would report a higher write throughput by starving large merges,
and LevelDB's merge scheduler would report a higher but unsustainable write throughput by dynamically adjusting the shape of the LSM-tree.
Based on our findings, we argue that in addition to the testing phase, used by existing LSM research,
an extra running phase must be performed to evaluate the usability of the measured maximum write throughput.
Moreover, the write latency must be measured properly due to queuing.
One solution is to use the proposed two-phase evaluation approach to evaluate the resulting write latencies under high utilization,
where the arrival rate is close to the processing rate.

\textbf{Merge scheduling is critical to minimizing write stalls.}
Throughout our evaluation of various LSM-tree designs, including bLSM~\cite{blsm2012}, full merges, and partitioned merges,
we have seen that merge scheduling has a critical impact on write stalls. 
Comparing these LSM-tree designs in general depends on many factors and is beyond the scope of this paper;
here we have focused on how to minimize write stalls for each LSM-tree design.

bLSM~\cite{blsm2012}, an instance of full merges, introduces a sophisticated spring-and-gear merge scheduler to bound the processing latency of LSM-trees. However, we found that bLSM still has large variances in its processing rate, leading to large write latencies under high arrival rates.
Among the three evaluated schedulers, namely single-threaded, fair, and greedy, the single-threaded scheduler should not be used in practical systems due to the long stalls caused by large merges.
The fair scheduler should be used when measuring the maximum throughput because it provides fairness to all merges.
The greedy scheduler should be used at runtime since it better minimizes the number of disk components,
both reducing write stalls and improving query performance.
Moreover, as an important design choice, global component constraints better minimizes write stalls.

Partitioned merges simplify merge scheduling by breaking large merges into many smaller ones.
However, we found a new problem that the measured maximum write throughput of LevelDB is unsustainable because it dynamically adjusts the size ratios under write-intensive workloads. After fixing this problem, a single-threaded scheduler with a small partition size, as used by LevelDB,
is sufficient for delivering low write latencies under high utilization.
However, fixing this problem reduced the maximum write throughput of LevelDB roughly one-third in our evaluation.

For both full and partitioned merges, processing writes as quickly as possible better minimizes write latencies.
Finally, with proper merge scheduling, all LSM-tree designs can indeed minimize write stalls by delivering low write latencies under high utilizations.

\section{Conclusion}
\label{sec:conclusion}
In this paper, we have studied and evaluated
the write stall problem for various LSM-tree designs.
We first proposed a two-phase approach to use in evaluating
the impact of write stalls on percentile write latencies using a combination of closed and open system testing models.
We then {identified} and explored the design choices for LSM merge schedulers.
For full merges, we proposed a greedy scheduler that minimizes write stalls.
For partitioned merges, we found that a single-threaded scheduler is sufficient to provide a stable write throughput
but that the maximum write throughput must be measured {properly}.
Based on these findings, we have shown that performance variance must be considered together with write throughput
to ensure the actual usability of the measured throughput.

\noindent
\textbf{Acknowledgments.}
We thank Neal Young, Dongxu Zhao, and anonymous reviewers for their helpful feedback on the theorems in this paper.
This work has been supported by NSF awards CNS-1305430, IIS-1447720, IIS-1838248, and CNS-1925610 along with industrial support from Amazon, Google, and Microsoft and support from the Donald Bren Foundation (via a Bren Chair).

\balance

\bibliographystyle{abbrv}
\bibliography{../lsm-survey/lsm}

\begin{thebibliography}{10}

\bibitem{asterixdb-web}
{AsterixDB}.
\newblock \url{https://asterixdb.apache.org/}.

\bibitem{cassandra}
Cassandra.
\newblock \url{http://cassandra.apache.org/}.

\bibitem{compaction-stall2017}
Compaction stalls: something to make better in {RocksDB}.
\newblock
  \url{http://smalldatum.blogspot.com/2017/01/compaction-stalls-something-to-make.html}.

\bibitem{hbase}
{HBase}.
\newblock \url{https://hbase.apache.org/}.

\bibitem{leveldb}
{LevelDB}.
\newblock \url{http://leveldb.org/}.

\bibitem{blsmgithub}
{Read- and latency-optimized log structured merge tree}.
\newblock \url{https://github.com/sears/bLSM/}.

\bibitem{rocksdb}
{RocksDB}.
\newblock \url{http://rocksdb.org/}.

\bibitem{sylladb}
{SyllaDB}.
\newblock \url{https://www.scylladb.com/}.

\bibitem{tarantool}
{Tarantool}.
\newblock \url{https://www.tarantool.io/}.

\bibitem{tpcc}
{TPC-C}.
\newblock \url{http://www.tpc.org/tpcc/}.

\bibitem{wiredtiger}
{WiredTiger}.
\newblock \url{http://www.wiredtiger.com/}.

\bibitem{ycsb-bug}
{YCSB} change log.
\newblock
  \url{https://github.com/brianfrankcooper/YCSB/blob/master/core/CHANGES.md}.

\bibitem{compaction2015}
M.~Y. Ahmad and B.~Kemme.
\newblock Compaction management in distributed key-value datastores.
\newblock {\em PVLDB}, 8(8):850--861, 2015.

\bibitem{asterixdb2014}
S.~Alsubaiee et~al.
\newblock {AsterixDB}: A scalable, open source {BDMS}.
\newblock {\em PVLDB}, 7(14):1905--1916, 2014.

\bibitem{piql2011}
M.~Armbrust et~al.
\newblock Piql: Success-tolerant query processing in the cloud.
\newblock {\em PVLDB}, 5(3):181--192, 2011.

\bibitem{scale-indep2013}
M.~Armbrust et~al.
\newblock Generalized scale independence through incremental precomputation.
\newblock In {\em ACM SIGMOD}, pages 625--636, 2013.

\bibitem{flodb2017}
O.~Balmau et~al.
\newblock {FloDB}: Unlocking memory in persistent key-value stores.
\newblock In {\em European Conference on Computer Systems (EuroSys)}, pages
  80--94, 2017.

\bibitem{triad2017}
O.~Balmau et~al.
\newblock {TRIAD}: Creating synergies between memory, disk and log in log
  structured key-value stores.
\newblock In {\em USENIX Annual Technical Conference (ATC)}, pages 363--375,
  2017.

\bibitem{bloom-filter1970}
B.~H. Bloom.
\newblock Space/time trade-offs in hash coding with allowable errors.
\newblock {\em CACM}, 13(7):422--426, July 1970.

\bibitem{join-predictable2009}
G.~Candea et~al.
\newblock A scalable, predictable join operator for highly concurrent data
  warehouses.
\newblock {\em PVLDB}, 2(1):277--288, 2009.

\bibitem{storage-var2017}
Z.~Cao et~al.
\newblock On the performance variation in modern storage stacks.
\newblock In {\em USENIX Conference on File and Storage Technologies (FAST)},
  pages 329--343, 2017.

\bibitem{asterixdb2019}
M.~J. Carey.
\newblock {AsterixDB} mid-flight: A case study in building systems in academia.
\newblock In {\em ICDE}, pages 1--12, 2019.

\bibitem{bigtable}
F.~Chang et~al.
\newblock Bigtable: A distributed storage system for structured data.
\newblock {\em ACM TOCS}, 26(2):4:1--4:26, 2008.

\bibitem{param-query-var2010}
S.~Chaudhuri et~al.
\newblock Variance aware optimization of parameterized queries.
\newblock In {\em ACM SIGMOD}, pages 531--542. ACM, 2010.

\bibitem{ycsb2010}
B.~F. Cooper et~al.
\newblock Benchmarking cloud serving systems with {YCSB}.
\newblock In {\em ACM SoCC}, pages 143--154, 2010.

\bibitem{monkey2017}
N.~Dayan et~al.
\newblock Monkey: Optimal navigable key-value store.
\newblock In {\em ACM SIGMOD}, pages 79--94, 2017.

\bibitem{monkey-tods}
N.~Dayan et~al.
\newblock Optimal {Bloom} filters and adaptive merging for {LSM}-trees.
\newblock {\em ACM TODS}, 43(4):16:1--16:48, Dec. 2018.

\bibitem{dostoevsky2018}
N.~Dayan and S.~Idreos.
\newblock Dostoevsky: Better space-time trade-offs for {LSM}-tree based
  key-value stores via adaptive removal of superfluous merging.
\newblock In {\em ACM SIGMOD}, pages 505--520, 2018.

\bibitem{lsm-bush}
N.~Dayan and S.~Idreos.
\newblock The log-structured merge-bush \& the wacky continuum.
\newblock In {\em ACM SIGMOD}, pages 449--466, 2019.

\bibitem{google-latency}
J.~Dean and L.~A. Barroso.
\newblock The tail at scale.
\newblock {\em CACM}, 56:74--80, 2013.

\bibitem{rocksdb-space2017}
S.~Dong et~al.
\newblock Optimizing space amplification in {RocksDB}.
\newblock In {\em CIDR}, volume~3, page~3, 2017.

\bibitem{asterixdb-feed2015}
R.~Grover and M.~J. Carey.
\newblock Data ingestion in {AsterixDB}.
\newblock In {\em EDBT}, pages 605--616, 2015.

\bibitem{queuing2013}
M.~Harchol-Balter.
\newblock {\em Performance modeling and design of computer systems: queueing
  theory in action}.
\newblock Cambridge University Press, 2013.

\bibitem{xengine}
G.~Huang et~al.
\newblock X-{Engine}: An optimized storage engine for large-scale {E-commerce}
  transaction processing.
\newblock In {\em ACM SIGMOD}, pages 651--665, 2019.

\bibitem{semantic-profiling2017}
J.~Huang et~al.
\newblock Statistical analysis of latency through semantic profiling.
\newblock In {\em European Conference on Computer Systems (EuroSys)}, pages
  64--79, 2017.

\bibitem{perf-predictability2018}
J.~Huang et~al.
\newblock A top-down approach to achieving performance predictability in
  database systems.
\newblock In {\em ACM SIGMOD}, pages 745--758, 2017.

\bibitem{fd-tree2010}
Y.~Li et~al.
\newblock Tree indexing on solid state drives.
\newblock {\em PVLDB}, 3(1-2):1195--1206, 2010.

\bibitem{hashkv2019}
Y.~Li et~al.
\newblock Enabling efficient updates in {KV} storage via hashing: Design and
  performance evaluation.
\newblock {\em ACM Transactions on Storage (TOS)}, 15(3):20, 2019.

\bibitem{lsm-model2016}
H.~Lim et~al.
\newblock Towards accurate and fast evaluation of multi-stage log-structured
  designs.
\newblock In {\em USENIX Conference on File and Storage Technologies (FAST)},
  pages 149--166, 2016.

\bibitem{wisckey2017}
L.~Lu et~al.
\newblock {WiscKey}: Separating keys from values in {SSD}-conscious storage.
\newblock In {\em USENIX Conference on File and Storage Technologies (FAST)},
  pages 133--148, 2016.

\bibitem{lsm-storage2019}
C.~Luo and M.~J. Carey.
\newblock Efficient data ingestion and query processing for {LSM}-based storage
  systems.
\newblock {\em PVLDB}, 12(5):531--543, 2019.

\bibitem{lsm-survey}
C.~Luo and M.~J. Carey.
\newblock {LSM}-based storage techniques: a survey.
\newblock {\em The VLDB Journal}, 2019.

\bibitem{umzi2019}
C.~Luo et~al.
\newblock Umzi: Unified multi-zone indexing for large-scale {HTAP}.
\newblock In {\em EDBT}, pages 1--12, 2019.

\bibitem{sifrdb2018}
F.~Mei et~al.
\newblock {SifrDB}: A unified solution for write-optimized key-value stores in
  large datacenter.
\newblock In {\em ACM SoCC}, pages 477--489, 2018.

\bibitem{lsm1996}
P.~O'Neil et~al.
\newblock The log-structured merge-tree ({LSM}-tree).
\newblock {\em Acta Inf.}, 33(4):351--385, 1996.

\bibitem{secondary2018}
M.~A. Qader et~al.
\newblock A comparative study of secondary indexing techniques in {LSM}-based
  {NoSQL} databases.
\newblock In {\em ACM SIGMOD}, pages 551--566, 2018.

\bibitem{pebblesdb2017}
P.~Raju et~al.
\newblock {PebblesDB}: Building key-value stores using fragmented
  log-structured merge trees.
\newblock In {\em ACM SOSP}, pages 497--514, 2017.

\bibitem{blink2009}
V.~Raman et~al.
\newblock Constant-time query processing.
\newblock In {\em ICDE}, pages 60--69, 2008.

\bibitem{slimdb2017}
K.~Ren et~al.
\newblock {SlimDB}: A space-efficient key-value storage engine for semi-sorted
  data.
\newblock {\em PVLDB}, 10(13):2037--2048, 2017.

\bibitem{stasis2006}
R.~Sears and E.~Brewer.
\newblock Stasis: Flexible transactional storage.
\newblock In {\em Symposium on Operating Systems Design and Implementation
  (OSDI)}, pages 29--44, 2006.

\bibitem{blsm2012}
R.~Sears and R.~Ramakrishnan.
\newblock {bLSM}: A general purpose log structured merge tree.
\newblock In {\em ACM SIGMOD}, pages 217--228, 2012.

\bibitem{deli2015}
Y.~Tang et~al.
\newblock Deferred lightweight indexing for log-structured key-value stores.
\newblock In {\em International Symposium in Cluster, Cloud, and Grid Computing
  (CCGrid)}, pages 11--20, 2015.

\bibitem{lsbm2017}
D.~Teng et~al.
\newblock {LSbM}-tree: Re-enabling buffer caching in data management for mixed
  reads and writes.
\newblock In {\em IEEE International Conference on Distributed Computing
  Systems (ICDCS)}, pages 68--79, 2017.

\bibitem{partial-merge2017}
R.~Thonangi and J.~Yang.
\newblock On log-structured merge for solid-state drives.
\newblock In {\em ICDE}, pages 683--694, 2017.

\bibitem{perf-predictable2009}
P.~Unterbrunner et~al.
\newblock Predictable performance for unpredictable workloads.
\newblock {\em PVLDB}, 2(1):706--717, 2009.

\bibitem{idea2019}
X.~Wang and M.~J. Carey.
\newblock An {IDEA}: An ingestion framework for data enrichment in {AsterixDB}.
\newblock {\em PVLDB.}, 12(11):1485--1498, 2019.

\bibitem{lwc-tree2017}
T.~Yao et~al.
\newblock A light-weight compaction tree to reduce {I/O} amplification toward
  efficient key-value stores.
\newblock In {\em International Conference on Massive Storage Systems and
  Technology (MSST)}, 2017.

\bibitem{ldc2019}
C.~Yunpeng et~al.
\newblock {LDC:} a lower-level driven compaction method to optimize
  {SSD}-oriented key-value stores.
\newblock In {\em ICDE}, pages 722--733, 2019.

\bibitem{elastic-bf2018}
Y.~Zhang et~al.
\newblock {ElasticBF}: Fine-grained and elastic bloom filter towards efficient
  read for {LSM}-tree-based {KV} stores.
\newblock In {\em USENIX Workshop on Hot Topics in Storage and File Systems
  (HotStorage)}, 2018.

\end{thebibliography}

\longonly{	
\appendix
\label{appendix:proof}
\section{Proofs of Theorems}

\setcounter{theorem}{0}

\begin{theorem}
Given any data arrival process and any LSM-tree, processing writes as quickly as possible minimizes the latency of each write.
\end{theorem}

\begin{proof}
	Given an LSM-tree, consider two merge schedulers $S$ and $S'$ which only differ in that 
	$S$ may add arbitrary delays to writes to avoid write stalls while $S'$ processes writes as quickly as possible.
	Denote the total number of writes processed by $S$ and $S'$ at time instant $T$ as $W_T$ and $W'_T$ respectively.
	Since $S'$ processes writes as quickly as possible, we have $W_T \le W'_T$.
	In other words, given the same numbers of writes, $S'$ processes these writes no later than $S$.
	
	Consider the $i$-th write request that arrives at time instant $T_{a_i}$.
	Suppose this write is processed by $S$ and $S'$ at time instants $T_{p_i}$ and $T'_{p_i}$ respectively.
	Based on the analysis above, it is straightforward that $T_{p_i} \ge T'_{p_i}$.
	Thus, we have $T_{p_i} - T_{a_i} \ge T'_{p_i} - T_{a_i}$, which implies that $S'$ minimizes the latency of each write.
\end{proof}

\begin{theorem}
Given any set of merge operations that process the same number of disk components and any {I/O bandwidth} budget,
the greedy scheduler minimizes the number of disk components at any time instant.
\end{theorem}

\begin{proof}
	Let $S$ be an arbitrary merge scheduler and $S'$ be the greedy scheduler.
	Suppose there are $N$ merge operations in total and the initial time instant is $t_0$.
	Denote by $t_i$ and $t'_i$ the time instants when $S$ and $S'$ complete their $i$-th merge operation, respectively.
	Since all merge operations always process the same number of disk components,
	we only need to show that for any $i\in [1, N]$, $t_i\ge t'_i$ always holds. In other words, $S'$ completes each merge operation no later than $S$.
	
	Suppose there exists $i\in [1, N]$ s.t. $t_i < t'_i$.
	Denote by $|S_{\le i}|$ and $|S'_{\le i}|$ the total number of bytes read and written by $S$ and $S'$ up to the completion of the $i$-th merge operation.
	By the definition of the greedy scheduler $S'$, we have $|S_{\le i}|\ge |S'_{\le i}|$.
	Since $t_i < t'_i$, we further have $\frac{|S_{\le i}|}{t_i-t_0} > \frac{|S_{'\le i}|}{t'_i-t_0}$.
	This implies that the merge scheduler $S$ requires a larger {I/O bandwidth} budget than $S'$, which leads to a contradiction.
	Thus, for any $i\in [1, N]$, $t_i\le t'_i$ always holds, which proves that $S'$ minimizes the number of disk components over time.
\end{proof}

\begin{theorem}
Given any {I/O bandwidth} budget, no merge scheduler can minimize the number of disk components at any time instant for any data arrival process and any LSM-tree for a deterministic merge policy where all merge operations process the same number of disk components.
\end{theorem}

\begin{proof}
	In this proof, we will construct an example showing that no such merge scheduler can be designed.
	Consider a two-level LSM-tree with a tiering merge policy. The size ratio of this merge policy is set at 2.
	Suppose Level $1$, which is the last level, contains three disk components $D_1$, $D_2$, $D_3$
	and {Level} $0$ contains two disk components, $D_4$ and $D_5$.
	For simplicity, assume that no more writes will arrive.
	Initially, the merge policy creates two merge operations, namely the merge operation $M_{1-2}$ that processes $D_1$ and $D_2$
	and the merge operation $M_{4-5}$ that processes $D_4$ and $D_5$.
	Upon the completion of $M_{1-2}$, which produces a new disk component $D_{1-2}$, the merge policy will create a new merge operation $M_{1-3}$ that processes $D_{1-2}$ and $D_3$.
	We further denote the amount of {I/O bandwidth} required by each merge operation $M_{1-2}$, $M_{4-5}$, and $M_{1-3}$ as $|M_{1-2}|$, $|M_{4-5}|$, and $|M_{1-3}|$.
	Finally, we assume that $|M_{1-3}| < |M_{4-5}| < |M_{1-2}|$. This can happen if $D_2$ contains a large number of deleted keys against $D_1$ so that the merged disk component $D_{1-2}$ is very small.
	
	Suppose that the initial time instant is $t_0$ and let the given {I/O bandwidth} budget be $B$.
	Consider a merge scheduler $S$ that first executes $M_{4-5}$ and then $M_{1-2}$.
	At time instant $t_1 = t_0 + \frac{|M_{4-5}|}{B}$, $S$ completes $M_{4-5}$ and reduces the number of disk components by $1$.
	At time instant $t_2 = t_0 + \frac{|M_{4-5}|}{B} + \frac{|M_{1-2}|}{B}$, $S$ completes $M_{1-2}$.
	Consider another merge scheduler $S'$ that first executes $M_{1-2}$.
	At time instant $t'_1 = t_0+ \frac{|M_{1-2}|}{B}$, $S'$ completes $M_{1-2}$.
	Now the merge policy creates a new merge operation $M_{1-3}$, which is then executed by $S'$.
	At time instant $t'_2 = t_0 +  \frac{|M_{1-2}|}{B} + \frac{|M_{1-3}|}{B}$, $S'$ completes $M_{1-3}$.
	Based on the assumption $|M_{1-3}| < |M_{4-5}| < |M_{1-2}|$, it follows that $t_1 < t'_1$ and $t'_2 < t_2$.
	Suppose there exists a merge scheduler $S^*$ that minimizes the number of disk components over time.
	Then, $S^*$ must satisfy the following two constraints:
	(1) complete one merge operation no later than $t_1$;
	(2) complete two merge operations no later than $t'_2$.
	
	To satisfy constraint (1), $S^*$ must execute $M_{4-5}$ first.
	Then, $S^*$ must complete the second merge operation within time interval $t'_2 - t_1 =  \frac{|M_{1-2}|}{B} + \frac{|M_{1-3}|}{B} -  \frac{|M_{4-5}|}{B}$.
	Since $|M_{1-3}| < |M_{4-5}|$, we have $t'_2 - t_1 < \frac{|M_{1-2}|}{B}$.
	Thus, $S^*$ cannot satisfy constraint (2) by completing the second merge operation no later than $t'_2$
	because the only remaining merge operation $M_{1-2}$ takes time $\frac{|M_{1-2}|}{B}$ to finish.
	This leads to a contradiction that $S^*$ minimizes the number of disk components over time.
	Thus, we have constructed an example for which no such merge scheduler can be designed, which proves the theorem.
\end{proof}
}
	
\end{document}